\documentclass[letterpaper, 10 pt, conference]{ieeeconf}  % Comment this line out
\IEEEoverridecommandlockouts                              % This command is only
\usepackage{cite}
\usepackage{amsmath}
\usepackage{amssymb}
\usepackage{algorithmic}
\usepackage{graphicx}
\usepackage{textcomp}
\usepackage{xcolor}
\usepackage{url}
\usepackage[hidelinks]{hyperref}

\usepackage[english]{babel}
\usepackage{graphicx}
\def\BibTeX{{\rm B\kern-.05em{\sc i\kern-.025em b}\kern-.08em
    T\kern-.1667em\lower.7ex\hbox{E}\kern-.125emX}}
\newcounter{maxpages}
\renewcommand\thesubsection{\Alph{subsection}}

\usepackage{multicol}
\usepackage{xspace}
\usepackage{mathtools}
\usepackage{booktabs}
\usepackage{multirow}
\usepackage{tabularx}
\newcolumntype{C}{>{\centering\arraybackslash}X}
\newcolumntype{L}{>{\raggedright\arraybackslash}X}
\usepackage{siunitx}
\usepackage{listings}
\usepackage{float}
\usepackage{bm}         % Negritas en modo matemático
\usepackage[normalem]{ulem}

\usepackage{amsmath,amssymb,amsfonts}
\usepackage{tikz}
\usetikzlibrary{arrows.meta, positioning, fit}

\input{tikz.tex}

\usepackage{bm}

\newtheorem{assumption}{Assumption}
\newtheorem{lemma}{Lemma}
\newtheorem{theorem}{Theorem}
\newtheorem{corollary}{Corollary}
\newtheorem{prop}{Proposition}
\newtheorem{remark}{Remark}

\DeclareMathOperator{\tr}{tr}

\begin{document}

\title{Learning-Based Geometric Leader–Follower Control for\\ Cooperative Rigid-Payload Transport with Aerial Manipulators

}
%Clifford Algebra and Applications to Navigation Systems\\
%\begin{comment}
\author{
    Omayra Yago Nieto$^1$, Leonardo Colombo$^2$
   \thanks{$^1$ Omaya Yago Nieto is with Universidad Politécnica de Madrid, Spain. {\tt\small omayra.yago.nieto@alumnos.upm.es}}
    \thanks{$^2$ L. Colombo is with Centre for Automation and Robotics (CSIC-UPM), Ctra. M300 Campo Real, Km 0,200, Arganda
del Rey - 28500 Madrid, Spain.{\tt\small leonardo.colombo@csic.es}}
\thanks{The author acknowledge financial support from Grant PID2022-137909NB-C21 funded by MCIN/AEI/ 10.13039/501100011033. The research leading to these results was supported in part by iRoboCity2030-CM, Robótica Inteligente para Ciudades Sostenibles (TEC-2024/TEC-62), funded by the Programas de Actividades I+D en Tecnologías en la Comunidad de Madrid.}
}
%%

%\end{comment}

\maketitle
\thispagestyle{empty}
\pagestyle{empty}

\begin{abstract}
This paper presents a learning-based tracking control framework for cooperative
transport of a rigid payload by multiple aerial manipulators under rigid grasp
constraints. A unified geometric model is developed, yielding a coupled
agent--payload differential--algebraic system that explicitly captures contact
wrenches, payload dynamics, and internal force redundancy. A leader--follower
architecture is adopted in which a designated leader generates a desired payload
wrench based on geometric tracking errors, while the remaining agents
realize this wrench through constraint-consistent force allocation.

Unknown disturbances and modeling uncertainties are compensated using Gaussian
Process (GP) regression. High-probability bounds on the learning error are
explicitly incorporated into the control design, combining GP feedforward
compensation with geometric feedback. Lyapunov analysis establishes uniform
ultimate boundedness of the payload tracking errors with high probability, with
an ultimate bound that scales with the GP predictive uncertainty.
\end{abstract}

%\renewcommand{\IEEEkeywordsname}{Keywords} 
%\begin{IEEEkeywords}Integrated Navigation, Trident Quaternions, Unmanned Aerial Vehicles, Sensor Fusion.\end{IEEEkeywords}

\section{Introduction} 
\label{sec:Intro}

Cooperative aerial transportation of large or heavy payloads has attracted
significant attention in recent years due to its relevance to inspection,
construction, logistics, and disaster-response applications
\cite{mellinger2011minimum, ryll2016modeling}.
Compared to single-vehicle aerial manipulation, cooperative transport increases
payload capacity and robustness through redundancy, but introduces challenging
coordination requirements and tightly coupled agent--payload dynamics.
These challenges are amplified under rigid grasp constraints, where contact
forces, internal-force redundancy, and payload dynamics must be handled
consistently within the control architecture.

Early approaches relied on simplified load models (e.g., cable-suspended or
quasi-static payloads), enabling decentralized or consensus-based strategies
\cite{goodman2023geometric, michael2011cooperative}.
More recent works address rigid payloads grasped by multiple aerial vehicles,
which require explicit wrench allocation and constraint handling
\cite{ollero2021past, wang2018cooperative}.
Geometric control on $SE(3)$ is particularly effective in this setting,
as they provide singularity-free tracking laws for translation and attitude
\cite{lee2010geometric, sreenath2013geometric}.

Despite this progress, most controllers presume accurate knowledge of vehicle
and payload dynamics.
In practice, cooperative aerial manipulation is affected by model mismatch,
aerodynamic disturbances, and uncertain or time-varying payload properties,
which can degrade tracking performance and jeopardize grasp consistency.
Learning-based augmentation is therefore appealing to compensate unmodeled
effects while retaining analytic guarantees.
Gaussian Process (GP) regression is especially attractive due to its
nonparametric modeling capability and uncertainty quantification
\cite{rasmussen2006gaussian, berkenkamp2017safe, beckers2021onlinem, beckers2021online},
but extending GP-based guarantees to rigidly grasped cooperative transport is
nontrivial because of holonomic constraints, internal-force redundancy, and
agent--payload coupling.

This paper develops a learning-augmented geometric tracking framework for
cooperative transport of a rigid payload by multiple aerial manipulators under
rigid grasp constraints.
We adopt a hierarchical leader--follower architecture: a leader generates a
desired payload wrench from geometric tracking errors, while the remaining
agents realize this wrench through constraint-consistent allocation and regulate
internal forces.
Learning is incorporated via GP regression to compensate unknown disturbances
with explicit high-probability error bounds used directly in the control design
and stability analysis.

The paper provides: (i) a unified geometric model leading
to coupled agent--payload DAE dynamics with explicit contact wrenches and rigid
grasp constraints; (ii) a leader--follower wrench generation/allocation scheme
that handles internal-force redundancy under thrust-direction underactuation;
and (iii) a GP-augmented controller with a Lyapunov analysis establishing
high-probability uniform ultimate boundedness of the payload tracking errors.

The remainder of the paper is as follows. Section~II derives the geometric model and coupled
dynamics. Section~III summarizes the GP learning model and error bounds.
Section~IV presents the controller and stability analysis. Section~V reports
simulation results, and Section~VI concludes.

%\begin{figure}[h!]
 %   \centering
 % \includegraphics[width=\columnwidth]{into.png}
  %  \caption{Cooperative aerial manipulation with multiple thrust-based UAVs equipped with robotic manipulators transporting a common rigid payload in a leader--follower configuration.}
   % \label{fig:intro}
%\end{figure}

\section{System Modeling: Single and Cooperative Aerial Manipulators}
\label{sec:ps}

This section develops a unified geometric model, from a single aerial manipulator to multi-agent and rigid-payload coupled dynamics. We first present the multibody Lie-group formulation, then derive the reduced thrust-based translational dynamics used for learning and control, and finally introduce the cooperative and payload-coupled extensions.

%Let $R\in \SO$ be a rotation matrix, $p\in\R^3$ a position vector, and $e_3 := [0\;\;0\;\;1]^\top$.
%We denote by $\hat{\cdot}:\R^3\to \mathfrak{so}(3)$ the standard hat map such that $\hat{\omega}v = \omega \times v$.

\subsection{Aerial manipulator as a multibody system.}
\label{subsec:geom_prelim}

%We consider a multirotor UAV consisting of $n{+}1$ interconnected rigid bodies arranged in a tree structure %(Fig.~\ref{fig1}).

%\begin{figure}[H]
 %   \centering
  %  \includegraphics[width=0.3\textwidth]{blue_to_gray.png}
   % \caption{Multi-body aerial robot with manipulator arms.}
   % \label{fig1}
%\end{figure}

We consider a multirotor UAV consisting of $n{+}1$ interconnected rigid bodies arranged in a tree structure. Let $R\in SO(3)$ and denote by $\mathfrak{so}(3)$ its Lie algebra. We use the standard hat/vee maps
$\hat{(\cdot)}:\mathbb R^3\to\mathfrak{so}(3)$ and $(\cdot)^\vee:\mathfrak{so}(3)\to\mathbb R^3$.
The configuration of the $i$th body is $g_i\in SE(3)$, written as
$g_i=\begin{pmatrix}R_i & \mathbf x_i\\ 0&1\end{pmatrix}$, where $R_i\in SO(3)$ and
$\mathbf x_i\in\mathbb R^3$ are the attitude and center-of-mass position.
The body twist is $\xi_i\in\mathfrak{se}(3)$ with matrix form
$\xi_i=\begin{pmatrix}\hat\omega_i & \mathbf v_i\\ 0&0\end{pmatrix}$, where
$\omega_i,\mathbf v_i\in\mathbb R^3$.

The manipulator configuration is described by joint coordinates $\mathbf r\in\mathbb R^{n_r}$ and joint torques
$\boldsymbol{\tau}_r\in\mathbb R^{n_r}$, where $n_r$ is the number of manipulator joint DoF. We use base variables
$q_0=(R_0,\mathbf x_0,\mathbf r)$ and $\xi_0=(\omega_0,\mathbf v_0,\dot{\mathbf r})$, where $\mathbf v_0=\dot{\mathbf x}_0$.
Let $g_i=g_0\,g_{0i}(\mathbf r)$ denote the relative transform from the base to body $i$.
The multibody kinetic energy induces a configuration-dependent mass matrix $\mathbf M_0(\mathbf r)$ in base coordinates
(see, e.g.,~\cite{marin}), which we partition as
\begin{equation}
\mathbf M_0(\mathbf r)=
\begin{bmatrix}
\mathbb J(\mathbf r) & \mathbf C(\mathbf r)^{\!\top} & \mathbf M_{\omega\dot r}(\mathbf r)\\
\mathbf C(\mathbf r) & m I_3 & \mathbf M_{v\dot r}(\mathbf r)\\
\mathbf M_{\omega\dot r}(\mathbf r)^{\!\top} & \mathbf M_{v\dot r}(\mathbf r)^{\!\top} & \mathbf M_{\dot r\dot r}(\mathbf r)
\end{bmatrix},
\label{eq:M0_partition}
\end{equation}
where $m:=\sum_{i=0}^{n} m_i$ is the total mass.
\smallskip

% For each body $i\in\{0,\dots,n\}$, let $g_i\in\SE(3)$ denote its pose, and assume the relative transforms satisfy
% \begin{equation}
% g_i = g_0\, g_{0i}(r),\qquad i=1,\dots,n.
% \label{eq:relative_transforms}
% \end{equation}
% Let $\xi_0=(\omega,v,\dot r)$ denote the base body twist (angular velocity $\omega\in\R^3$, body linear velocity $v\in\R^3$) augmented with joint rates. The pose inertia of body $i$ is
% \begin{equation}
% \mathbb{I}_i=\begin{bmatrix}\mathbb{J}_i & 0\\0 & m_i I_3\end{bmatrix},
% \end{equation}
% where $\mathbb{J}_i$ is the rotational inertia and $m_i$ its mass. The total mass is $m=\sum_{i=0}^n m_i$.

Following~\cite{marin}, define the instantaneous center-of-mass position $\displaystyle{\mathbf x \;:=\; \frac{1}{m}\sum_{i=0}^{n} m_i\,\mathbf x_i \in\mathbb R^3}$, and perform the standard change of variables $R:=R_0$, $\omega:=\omega_0$,
\begin{equation}
\mathbf v \;:=\; \mathbf v_0 + \frac{1}{m}\Big(\mathbf M_{v\dot r}(\mathbf r)\dot{\mathbf r} + \mathbf C(\mathbf r)\omega\Big),
\label{eq:vbar_def}
\end{equation}
which diagonalizes the translational block of the mass matrix. We use the reduced state
$\bar{\mathbf q} \;:=\; \big(R,\mathbf x,\mathbf r,\omega,\mathbf v,\dot{\mathbf r}\big)$.

Let $\mathbf e_3=(0,0,1)^\top$ and let $u>0$ denote the lift magnitude along $R\mathbf e_3$, while
$\boldsymbol{\tau}\in\mathbb R^3$ is the body torque acting on the base. The nominal dynamics can be written as
\begin{align}
m\ddot{\mathbf x} &= m\mathbf a_g + R\mathbf e_3\,u, 
\qquad \dot R = R\hat\omega,
\label{eq:nom_trans_rot}\\
\begin{bmatrix}\omega\\ \dot{\mathbf r}\end{bmatrix}
&= \bar{\mathbf M}(\mathbf r)^{-1}\begin{bmatrix}\boldsymbol{\mu}\\ \boldsymbol{\nu}\end{bmatrix},
\label{eq:nom_kin_coup}\\
\begin{bmatrix}\dot{\boldsymbol{\mu}}\\ \dot{\boldsymbol{\nu}}\end{bmatrix}
&=
\begin{bmatrix}
\boldsymbol{\mu}\times\omega\\[1mm]
\frac{1}{2}\bar{\boldsymbol{\xi}}^{\!\top}\,\partial \bar{\mathbf M}(\mathbf r)\,\bar{\boldsymbol{\xi}}
\end{bmatrix}
+
\begin{bmatrix}
\boldsymbol{\tau} - \mathbf C(\mathbf r)^{\!\top} \mathbf e_3\,u/m\\[1mm]
\boldsymbol{\tau}_r - \mathbf M_{v\dot r}(\mathbf r)^{\!\top}\mathbf e_3\,u/m
\end{bmatrix},
\label{eq:nom_momentum}
\end{align}
where $\bar{\boldsymbol{\xi}}:=(\omega,\dot{\mathbf r})$ and the reduced mass matrix is
\begin{equation}
    \bar{\bf{M}}=\begin{bmatrix}
    \mathbb{J}-C^{T}C/m & \mathbf{M}_{\omega\dot{r}}-C^{T}\mathbf{M}_{v\dot{r}}/m\\
   \mathbf{M}_{\omega\dot{r}}^{T}-\mathbf{M}_{v\dot{r}}^{T}C/m& \mathbf{M}_{\dot{r}\dot{r}}-\mathbf{M}_{v\dot{r}}^{T}\mathbf{M}_{v\dot{r}}/m
    \end{bmatrix}.
\end{equation}
Here $\mathbf a_g\in\mathbb R^3$ is the gravitational acceleration expressed in the inertial frame.

In realistic scenarios, \eqref{eq:nom_trans_rot}--\eqref{eq:nom_momentum} are affected by unmodeled and uncertain
effects (e.g., aerodynamic forces, actuator dynamics, friction, payload variations, and couplings induced by manipulator motion).
We represent these effects by unknown smooth functions of the reduced state $\bar{\mathbf q}$:
\[
\mathbf f_{uk}^{x}:\mathcal X\to\mathbb R^{3},\,
\mathbf f_{uk}^{\omega}:\mathcal X\to\mathbb R^{3},\,
\mathbf f_{uk}^{\dot r}:\mathcal X\to\mathbb R^{n_r},
\] where $\mathcal X\subseteq SO(3)\times\mathbb R^3\times\mathbb R^{n_r}\times\mathbb R^3\times\mathbb R^3\times\mathbb R^{n_r}$. The uncertain dynamics then read
\begin{align}
m\ddot{\mathbf x} &= m\mathbf a_g + R\mathbf e_3\,u + \mathbf f_{uk}^{x}(\bar{\mathbf q}),
\qquad \dot R = R\hat\omega,
\label{eq:sys_unc_1}\\
\begin{bmatrix}\omega\\ \dot{\mathbf r}\end{bmatrix}
&= \bar{\mathbf M}(\mathbf r)^{-1}\begin{bmatrix}\boldsymbol{\mu}\\ \boldsymbol{\nu}\end{bmatrix},
\label{eq:sys_unc_2}\\
\begin{bmatrix}\dot{\boldsymbol{\mu}}\\ \dot{\boldsymbol{\nu}}\end{bmatrix}
&=
\begin{bmatrix}
\boldsymbol{\mu}\times\omega\\[1mm]
\frac{1}{2}\bar{\boldsymbol{\xi}}^{\!\top}\,\partial \bar{\mathbf M}(\mathbf r)\,\bar{\boldsymbol{\xi}}
\end{bmatrix}
+
\begin{bmatrix}
\boldsymbol{\tau} - \mathbf C(\mathbf r)^{\!\top} \mathbf e_3\,u/m\\[1mm]
\boldsymbol{\tau}_r - \mathbf M_{v\dot r}(\mathbf r)^{\!\top}\mathbf e_3\,u/m
\end{bmatrix}
\\&+
\begin{bmatrix}
\mathbf f_{uk}^{\omega}(\bar{\mathbf q})\\
\mathbf f_{uk}^{\dot r}(\bar{\mathbf q})\nonumber
\end{bmatrix}.
\label{eq:sys_unc_3}
\end{align}

\subsection{Cooperative aerial manipulation system}
\label{subsec:dae_model}

We now write the dynamics of the cooperative aerial manipulation system in a unified
differential--algebraic equation (DAE) form, explicitly accounting for the rigid payload,
the holonomic grasp constraints, and the contact wrenches generated by each aerial manipulator.

For each aerial manipulator $j\in\mathcal \{1,\dots,N\}$, we consider the reduced state $\bar q_j=(\mathbf q_j,\xi_j)=(R_j,x_j,\mathbf r_j,\omega_j,v_j,\dot{\mathbf r}_j)$, as defined in Section II-\ref{subsec:geom_prelim}.
The payload state is given by $q_L:=(p_L,v_L,R_L,\omega_L)$, where $v_L:=\dot p_L$.
The complete system state is $z := \Big(\{\bar q_j\}_{j=1}^N,\; q_L\Big)$.

Each UAV agent $j\in\{1,\dots,N\}$ is equipped with $n$ manipulator arms.
Each arm $b\in\{1,\dots,n\}$ establishes a single rigid contact with the payload.
The corresponding contact forces are denoted by $\lambda_{j,b}\in\mathbb R^3$, and we define the stacked
contact force vector $\lambda :=
\begin{bmatrix}
\lambda_{1,1}^\top & \cdots & \lambda_{1,n}^\top &
\cdots &
\lambda_{N,1}^\top & \cdots & \lambda_{N,n}^\top
\end{bmatrix}^{\top}
\in\mathbb R^{3Nn}$.

Let $c^L_{j,b}\in\mathbb R^3$ be the fixed attachment point expressed in the payload frame,
and let $p_{j,b}(\mathbf q_j)$ denote the inertial position of the corresponding contact
point on UAV $j$ (arm $b$).
The rigid grasp constraints, for $b=1,\dots,n$ and $j=1,\dots,N$, are
\begin{equation}\label{eq:dae_constraints}
\phi_{j,b}(z) :=
p_{j,b}(\mathbf q_j) - \big(p_L + R_L c^L_{j,b}\big) = 0.
\end{equation}

Differentiating twice yields the acceleration-level compatibility conditions used implicitly
in the DAE formulation.

Each aerial manipulator obeys the thrust-based reduced dynamics augmented by contact forces:
\begin{align}
\dot R_j &= R_j\hat\omega_j,\label{eq:dae_agent_R}\\
m_j\ddot x_j &= m_j\mathbf a_g + R_j\mathbf e\,u_j
+ \mathbf f^{x}_{uk,j}(\mathbf q_j,\xi_j)
+ f^{c}_{x,j}(\lambda),\label{eq:dae_agent_x}\\
\begin{bmatrix}\omega_j\\ \dot{\mathbf r}_j\end{bmatrix}
&=\bar{\mathbf M}_j^{-1}(\mathbf r_j)
\begin{bmatrix}\mu_j\\ \nu_j\end{bmatrix},\label{eq:dae_agent_kin}\\
\begin{bmatrix}\dot\mu_j\\ \dot\nu_j\end{bmatrix}
&=
\begin{bmatrix}
\mu_j\times\omega_j\\
\frac12\bar\xi_j^\top\partial\bar{\mathbf M}_j(\mathbf r_j)\bar\xi_j
\end{bmatrix}
+
\begin{bmatrix}
\tau_j - C_j^\top\mathbf e\,u_j/m_j\\
\tau_{r,j}-\mathbf M_{j,v\dot{r}}^\top\mathbf e\,u_j/m_j
\end{bmatrix}\nonumber\\
&\quad+
\begin{bmatrix}
\mathbf f^{\omega}_{uk,j}(\mathbf q_j,\xi_j)\\
\mathbf f^{\dot r}_{uk,j}(\mathbf q_j,\xi_j)
\end{bmatrix}
-\sum_{b=1}^{n}
J_{c,j,b}(\mathbf q_j)^\top\,\lambda_{j,b},
\label{eq:dae_agent_mom}
\end{align}
where $J_{c,j,b}(\mathbf q_j)$ denotes the contact Jacobian associated with contact $(j,b)$.
The term $f^{c}_{x,j}(\lambda)$ represents the contribution of the contact forces to the
translational channel, which may be omitted if all contact effects are accounted for in
\eqref{eq:dae_agent_mom}.

The payload kinematics is given by $\dot p_L = v_L$, where $v_L\in\mathbb R^3$ is the payload linear velocity. The payload is driven by the net contact wrench and is subject to unmodeled effects:
\begin{align}
m_L\dot v_L &=
\sum_{j=1}^{N}\sum_{b=1}^{n} \lambda_{j,b}
+ f_{g,L}
\nonumber\\&
+ \mathbf f^{p}_{uk,L}(p_L,v_L,R_L,\omega_L),
\label{eq:dae_payload_trans}\\
J_L\dot\omega_L + \omega_L\times(J_L\omega_L)
&=
\sum_{j=1}^{N}\sum_{b=1}^{n}
c^L_{j,b}\times\big(R_L^\top\lambda_{j,b}\big)\nonumber
\\&+ \mathbf f^{\omega}_{uk,L}(p_L,v_L,R_L,\omega_L).
\label{eq:dae_payload_rot}
\end{align}

Defining the payload wrench $W_L=[F_L;\tau_L]\in\mathbb R^6$, the contact forces are mapped
to the payload dynamics via the grasp matrix $W_L = G(R_L)\,\lambda$
\begin{align*}
G(R_L)&=
\begin{bmatrix}
I_3 & \cdots & I_3\\
\widehat{R_L c^L_{1,1}} & \cdots & \widehat{R_L c^L_{N,n}}
\end{bmatrix}
\in\mathbb R^{6\times 3Nn},
\end{align*}
where the columns are ordered with the stacking of $\lambda$.

We stack the contact forces as $\lambda\in\mathbb R^{3Nn}$ and the holonomic constraints as
$\Phi(z)\in\mathbb R^{3Nn}$.
Assuming identical manipulators $\mathbf r_j\in\mathbb R^{n_r}$ for all $j$, the differential
state dimension is $\dim(z)= N(12+2n_r)+12 \;=\; 12(N+1)+2Nn_r$.

Equations \eqref{eq:dae_agent_R}--\eqref{eq:dae_payload_rot}, together with the algebraic
constraints \eqref{eq:dae_constraints}, define a coupled DAE with $12(N+1)+2Nn_r$ differential
states and $3Nn$ algebraic variables $\lambda$ subject to $3Nn$ holonomic constraints of the form
\begin{equation}\label{eq:dae_compact}
\begin{aligned}
\dot z &= F(z,u,\lambda) + F_{uk}(z),\quad 0= \Phi(z),\\
\end{aligned}
\end{equation}
where $\Phi(z)$ stacks the holonomic grasp constraints. Here $u\in\mathcal U$ is the stacked vector of differential control inputs of all agents,
e.g., $u=\mathrm{col}(u_1,\tau_1,\tau_{r,1},\ldots,u_N,\tau_N,\tau_{r,N})$.

%This formulation serves as the basis for the leader--follower wrench generation and force
%allocation strategy, as well as for learning-based compensation of uncertainties acting on
%both the aerial manipulators and the payload.

\section{Learning with Gaussian Processes}
\label{sec:learning}

Building on the learning-based tracking framework in~\cite{YagoColomboArms}, we consider cooperative transport of a rigid payload
by a team of $N$ UAV agents. The coupled closed-loop model in Section~II-\ref{subsec:dae_model} is affected by unknown and unmodeled
effects including aerodynamic disturbances, actuator dynamics, manipulator--payload coupling errors, and payload parameter mismatch.
We model these effects as additive, state-dependent, and time-invariant generalized disturbance terms on the time scale of interest.

% ===========================
% Unknown maps and predictors
% ===========================
For each UAV agent $j\in\{1,\dots,N\}$ we consider unknown maps
$\mathbf f_{uk,j}^{x}$, $\mathbf f_{uk,j}^{\omega}$, $\mathbf f_{uk,j}^{\dot r}:\; SE(3)\times\mathfrak{se}(3)\rightarrow\mathbb R^3$,
acting on the translational, rotational, and joint/momentum dynamics of agent $j$, evaluated at the reduced state
$\bar q_j=(\mathbf q_j,\xi_j)=(R_j,x_j,\mathbf r_j,\omega_j,v_j,\dot{\mathbf r}_j)$ (with $\mathbf r_j\in\mathbb R^{n_r}$ for all $j$).
In addition, the payload dynamics include unknown force/moment terms
$\mathbf f_{uk,L}^{p}$, $\mathbf f_{uk,L}^{\omega}:\; SE(3)\times\mathfrak{se}(3)\rightarrow\mathbb R^3$,
evaluated at the payload state $q_L:=(p_L,v_L,R_L,\omega_L)$.

We denote by
\[
\hat{\mathbf f}_{uk,j}(\bar q_j):=
\begin{bmatrix}
\hat{\mathbf f}_{uk,j}^{x}(\bar q_j)\\
\hat{\mathbf f}_{uk,j}^{\omega}(\bar q_j)\\
\hat{\mathbf f}_{uk,j}^{\dot r}(\bar q_j)
\end{bmatrix},\quad 
\hat{\mathbf f}_{uk,L}(q_L):=
\begin{bmatrix}
\hat{\mathbf f}_{uk,L}^{p}(q_L)\\
\hat{\mathbf f}_{uk,L}^{\omega}(q_L)
\end{bmatrix}
\]
the learned predictors used for feedforward compensation. In the GP setting below,
$\hat{\mathbf f}_{uk,j}\in\mathbb R^9$ and $\hat{\mathbf f}_{uk,L}\in\mathbb R^6$ are chosen as GP posterior means, with associated predictive covariance matrices
denoted $\Sigma_j(\cdot)$ and $\Sigma_L(\cdot)$.

% ===========================
% Online datasets (piecewise-constant updates)
% ===========================
We introduce time-varying datasets collected during execution. We consider piecewise-constant datasets over intervals
$t\in[t_k,t_{k+1})$ with $0=t_0<t_1<t_2<\cdots$. Let $\mathcal D_j(t)$ and $\mathcal D_L(t)$ denote the streaming datasets.
On each interval $[t_k,t_{k+1})$ we freeze the data at the update time by setting
$\mathcal D_{j,k}:=\mathcal D_j(t_k)$ and $\mathcal D_{L,k}:=\mathcal D_L(t_k)$.
The corresponding (time-indexed) predictors are denoted $\hat{\mathbf f}_{uk,j,k}$ and $\hat{\mathbf f}_{uk,L,k}$.

% ===========================
% Agent learning outputs (labels)
% ===========================
For each agent $j$, define the learning output
$y_j :=
\begin{bmatrix}
y^{x}_j &
y^{\omega}_j &
y^{\dot r}_j
\end{bmatrix}^{\top}
\in\mathbb R^9$ by rearranging the agent dynamics and subtracting estimated contact contributions.
Let $\hat\lambda_{j,b}\in\mathbb R^3$ denote the estimated contact forces at the attachment points $c^L_{j,b}$ (payload frame),
with $j\in\{1,\dots,N\}$ and $b\in\{1,\dots,n\}$, and let $\hat\lambda$ denote the corresponding stacked vector.
Then we set
\begin{align}
y^{x}_j &:= m_j\ddot{x}_j - m_j\mathbf a_g - R_j\mathbf e\, u_j - f^{c}_{x,j}(\hat\lambda), \nonumber\\
y^{\omega}_j &:= \dot\mu_j - \mu_j\times\omega_j
+ C_j^{\top}\mathbf e\, u_j/m_j \\&- \tau_j
+ \sum_{b=1}^{n} \big(J_{c,j,b}(\mathbf q_j)^\top \hat\lambda_{j,b}\big)_{(1:3)}, \nonumber\\
y^{\dot r}_j &:= \dot\nu_j
- \tfrac12 \bar\xi_j^{\top}\partial\bar{\mathbf M}_j(\mathbf r_j)\bar\xi_j
+ \mathbf M_{j,v\dot{r}}^{\top}\mathbf e\, u_j/m_j \\&- \tau_{r,j}
+ \sum_{b=1}^{n} \big(J_{c,j,b}(\mathbf q_j)^\top \hat\lambda_{j,b}\big)_{(4:3+n_r)}.
\label{eq:learning_output_components}
\end{align}
We assume the contact generalized-force term $J_{c,j,b}(\mathbf q_j)^\top\hat\lambda_{j,b}\in\mathbb R^{3+n_r}$ is stacked consistently
with $\begin{bmatrix}\mu_j\\ \nu_j\end{bmatrix}$, so that its first $3$ components act on the $\mu_j$-channel and the remaining $n_r$
components act on the $\nu_j$-channel.

Under the coupled DAE dynamics, these labels satisfy
$y^{x}_j=\mathbf f_{uk,j}^{x}$,
$y^{\omega}_j=\mathbf f_{uk,j}^{\omega}$,
$y^{\dot r}_j=\mathbf f_{uk,j}^{\dot r}$,
up to measurement noise and numerical differentiation errors.

% --- online dataset for agent j ---
We collect streaming data pairs $(\bar q_j(t),\tilde y_j(t))$, where $\tilde y_j=y_j+\eta_j$ and $\eta_j$ models measurement noise, and define
the time-varying dataset
\begin{equation}
\mathcal D_j(t)=\{(\bar q_j^{\{i\}},\,\tilde y_j^{\{i\}})\}_{i=1}^{N_j(t)}.
\label{eq:dataset_agent_online}
\end{equation}
On $t\in[t_k,t_{k+1})$ we freeze $\mathcal D_{j,k}:=\mathcal D_j(t)$ and compute the oracle
$\hat{\mathbf f}_{uk,j,k}(\cdot)=\mu_{j,k}(\cdot\,|\,\mathcal D_{j,k})$ as the GP posterior mean, with covariance $\Sigma_{j,k}(\cdot)$.

% ===========================
% Payload learning outputs (labels)
% ===========================

Define the payload learning output
$y_L :=
\begin{bmatrix}
y_L^{p} &
y_L^{\omega}
\end{bmatrix}^{\top}
\in\mathbb R^6$ by
\begin{align*}
y_L^{p}
&:= m_L\dot v_L - \sum_{j=1}^{N}\sum_{b=1}^{n} \hat\lambda_{j,b} - f_{g,L},\\
y_L^{\omega}
&:= J_L\dot\omega_L + \omega_L\times(J_L\omega_L)
- \sum_{j=1}^{N}\sum_{b=1}^{n}
c^L_{j,b}\times\!\big(R_L^\top\hat\lambda_{j,b}\big).
\end{align*}
Then, under the coupled DAE dynamics,
$y_L^{p}=\mathbf f_{uk,L}^{p}$ and $y_L^{\omega}=\mathbf f_{uk,L}^{\omega}$,
up to measurement noise and estimation errors (including $\hat\lambda$ and differentiation).

Let $\mathcal Q_L := \mathbb R^3\times\mathbb R^3\times SO(3)\times\mathbb R^3$ denote the payload state space,
with $q_L=(p_L,v_L,R_L,\omega_L)\in\mathcal Q_L$. We collect the time-varying payload dataset
\begin{equation}
\mathcal D_L(t)=\{(q_L^{\{i\}},\,\tilde y_L^{\{i\}})\}_{i=1}^{N_L(t)}
\label{eq:dataset_payload_online}
\end{equation}
where $\tilde y_L=y_L+\eta_L$.
On $t\in[t_k,t_{k+1})$ we freeze $\mathcal D_{L,k}:=\mathcal D_L(t)$ and compute the oracle
$\hat{\mathbf f}_{uk,L,k}(\cdot)=\mu_{L,k}(\cdot\,|\,\mathcal D_{L,k})$ as the GP posterior mean, with covariance $\Sigma_{L,k}(\cdot)$.

%========================================================
\begin{assumption}
\label{ass:dataset}
For each agent $j$ and for the payload, the corresponding training data sets
$\mathcal D_j(t)$ and $\mathcal D_L(t)$ are updated only finitely many times.
In particular, there exists a time $T_{\mathrm{end}}\in\mathbb{R}_{\ge 0}$ such that,
for all $t\ge T_{\mathrm{end}}$, $\mathcal D_j(t)=\mathcal D_j^{\mathrm{end}}$, $\forall j$, $\mathcal D_L(t)=\mathcal D_L^{\mathrm{end}}$.
\end{assumption}

Let $S_{\mathcal X}\subset \mathcal Z$ be a compact subset of the coupled agent--payload state space
$\mathcal Z := \mathcal X^N \times \mathcal Q_L$ containing the relevant closed-loop trajectories.
Assume that $\hat{\mathbf f}_{uk,j}$ and $\hat{\mathbf f}_{uk,L}$ are continuously differentiable with bounded derivatives on $S_{\mathcal X}$.
For the stability analysis, we require high-probability bounds on the learning error that depend on the confidence level.

\begin{assumption} \label{ass:oracle_bound}For any confidence levels $\delta_L\in(0,1)$ and
$\delta_j\in(0,1)$, $j\in\{1,\dots,N\}$, there exist bounded functions
$\bar\rho_j(\cdot;\delta_j):\mathcal{X}\to\mathbb{R}_{\ge 0}$ and
$\bar\rho_L(\cdot;\delta_L):Q_L\to\mathbb{R}_{\ge 0}$ such that, for all $z\in S_\mathcal{X}$,
\begin{align}
\mathbb{P}\!\left(\|f_{\mathrm{uk},j}(\bar q_j)-\hat f_{\mathrm{uk},j}(\bar q_j)\|
\le \bar\rho_j(\bar q_j;\delta_j)\right) \ge \delta_j,\ \ \forall j, \label{eq:ass2_agent}
\\
\mathbb{P}\!\left(\|f_{\mathrm{uk},L}(q_L)-\hat f_{\mathrm{uk},L}(q_L)\|
\le \bar\rho_L(q_L;\delta_L)\right) \ge \delta_L. \label{eq:ass2_payload}
\end{align}
\end{assumption}

\begin{remark}\label{rem1} For a single global confidence level $\delta\in(0,1)$, define $\delta_L := 1-\frac{1-\delta}{N+1}$ and $\delta_j := 1-\frac{1-\delta}{N+1}$, $j=1,\dots,N$. Then, by a union bound, the joint confidence event
$\Omega_\delta := \Omega_L(\delta_L)\cap\bigcap_{j=1}^N\Omega_j(\delta_j)$
satisfies $\mathbb{P}(\Omega_\delta)\ge \delta$.
(Any other split satisfying $(1-\delta_L)+\sum_{j=1}^N(1-\delta_j)\le 1-\delta$
is admissible.)\end{remark}

For each agent $j$, we model each scalar component of $\tilde y_j\in\mathbb R^9$ by an independent GP
with kernel $k_j(\cdot,\cdot)$ and mean $m_{j,i}(\cdot)$.
At an update time $t_k$, let the frozen dataset be $\mathcal D_{j,k}=\{(\bar q_j^{\{i\}},\tilde y_j^{\{i\}})\}_{i=1}^{N_{j,k}}$ and
stack the training inputs and outputs as $X_{j,k}:=\big[\bar q_j^{\{1\}},\ldots,\bar q_j^{\{N_{j,k}\}}\big]$, $Y_{j,k}^\top:=\big[\tilde y_j^{\{1\}},\ldots,\tilde y_j^{\{N_{j,k}\}}\big]$. Assume additive i.i.d.\ Gaussian measurement noise,
$\tilde y_j^{\{i\}} = y_j^{\{i\}} + \varepsilon_{j}^{\{i\}}$ with $\varepsilon_{j}^{\{i\}}\sim\mathcal N(0,\sigma_j^2 I_9)$.
For a query $\bar q_j^*$, the GP posterior predictive mean and variance for the $i$-th component are
\begin{align*}
\mu_{j,k,i}(\bar q_j^*)
&= m_{j,i}(\bar q_j^*)
+ k_j(\bar q_j^*,X_{j,k})^\top K_{j,k}^{-1}\!\left(Y_{j,k,:,i}-m_{j,i}(X_{j,k})\right),
\\
\sigma^2_{j,k,i}(\bar q_j^*)
&= k_j(\bar q_j^*,\bar q_j^*)
- k_j(\bar q_j^*,X_{j,k})^\top K_{j,k}^{-1}k_j(\bar q_j^*,X_{j,k}),
\end{align*}
where the Gram matrix $K_{j,k}\in\mathbb R^{N_{j,k}\times N_{j,k}}$ has entries $(K_{j,k})_{\ell m}=k_j(\bar q_j^{\{\ell\}},\bar q_j^{\{m\}})+\sigma_j^2\,\delta_{\ell m}$, and
$m_{j,i}(X_{j,k}):=[m_{j,i}(\bar q_j^{\{1\}}),\ldots,m_{j,i}(\bar q_j^{\{N_{j,k}\}})]^\top$.

We define the vector-valued GP posterior quantities
\begin{align*}
\mu_{j,k}(\bar q_j^*) :=&
\begin{bmatrix}
\mu_{j,k,1}(\bar q_j^*) & \cdots & \mu_{j,k,9}(\bar q_j^*)
\end{bmatrix}^{\top}\in\mathbb R^{9},\\
\Sigma_{j,k}(\bar q_j^*) :=&
\mathrm{diag}\!\left(\sigma^2_{j,k,1}(\bar q_j^*),\ldots,\sigma^2_{j,k,9}(\bar q_j^*)\right)\in\mathbb R^{9\times 9}.
\end{align*}
On $t\in[t_k,t_{k+1})$, we use $\hat{\mathbf f}_{uk,j,k}(\bar q_j^*)=\mu_{j,k}(\bar q_j^*)$ and $\Sigma_{j,k}(\bar q_j^*)$.

Similarly, each scalar component of $\tilde y_L\in\mathbb R^6$ is modeled by an independent GP with kernel $k_L(\cdot,\cdot)$,
mean $m_{L,i}(\cdot)$, and noise variance $\sigma_L^2$.
At update time $t_k$, let $\mathcal D_{L,k}=\{(q_L^{\{i\}},\tilde y_L^{\{i\}})\}_{i=1}^{N_{L,k}}$ and stack $X_{L,k}:=\big[q_L^{\{1\}},\ldots,q_L^{\{N_{L,k}\}}\big]$ and 
$Y_{L,k}^\top:=\big[\tilde y_L^{\{1\}},\ldots,\tilde y_L^{\{N_{L,k}\}}\big]$. Assume $\tilde y_L^{\{i\}} = y_L^{\{i\}} + \varepsilon_{L}^{\{i\}}$ with $\varepsilon_{L}^{\{i\}}\sim\mathcal N(0,\sigma_L^2 I_6)$.
For a query $q_L^*$, define $\mu_{L,k}(q_L^*)\in\mathbb R^6$ and $\Sigma_{L,k}(q_L^*)\in\mathbb R^{6\times 6}$ exactly as above
(with $j\mapsto L$ and dimension $9\mapsto 6$). On $t\in[t_k,t_{k+1})$, we use
$\hat{\mathbf f}_{uk,L,k}(q_L^*)=\mu_{L,k}(q_L^*)$ and $\Sigma_{L,k}(q_L^*)$.

\begin{remark}
The mean functions $m_{j,i}$ and $m_{L,i}$ may incorporate parametric identification of the unknown terms.
In the absence of prior knowledge, they are set to zero.
\end{remark}

%========================================================

To connect GP regression with the high-probability error bounds in Assumption~\ref{ass:oracle_bound},
we rely on standard results from reproducing kernel Hilbert space (RKHS) theory.
Given a positive definite kernel $k:\mathcal X\times\mathcal X\rightarrow\mathbb R$,
there exists a unique Hilbert space $\mathcal H_k$ of functions $f:\mathcal X\rightarrow\mathbb R$,
called the reproducing kernel Hilbert space associated with $k$,
such that the reproducing property
$f(x)=\langle f,k(x,\cdot)\rangle_{\mathcal H_k}$
holds for all $f\in\mathcal H_k$ and $x\in\mathcal X$~\cite{wahba1990spline,steinwart2008support}.

\begin{assumption}
\label{ass:rkhs}
For each UAV agent $j\in\{1,\dots,N\}$ and each scalar output component $i\in\{1,\dots,9\}$,
the unknown function $f_{uk,j,i}$ belongs to the RKHS $\mathcal H_{k_j}$ induced by the kernel $k_j$,
with uniformly bounded norm.
Similarly, for the payload-level unknown dynamics, each scalar component $f_{uk,L,i}$, $i\in\{1,\dots,6\}$,
belongs to the RKHS $\mathcal H_{k_L}$ induced by the kernel $k_L$.
That is, there exist constants $B_j>0$ and $B_L>0$ such that $\|f_{uk,j,i}\|_{\mathcal H_{k_j}} \le B_j$, $\forall\, i\in\{1,\dots,9\},\ \forall\, j$, and $\|f_{uk,L,i}\|_{\mathcal H_{k_L}} \le B_L$, $\forall\, i\in\{1,\dots,6\}$.
\end{assumption}

\begin{lemma}[Adapted from \cite{Srinivas2012}]
\label{lem:gp_error_bound}
Suppose Assumptions~1 and~3 hold, and that the GP models for the agents and the payload
are trained on the (frozen) data sets $\mathcal{D}^{\mathrm{end}}_j$ and $\mathcal{D}^{\mathrm{end}}_L$, respectively.
Let $\mu_j(\cdot),\Sigma_j(\cdot)$ and $\mu_L(\cdot),\Sigma_L(\cdot)$ denote the corresponding terminal GP posterior mean and covariance.
Then, for any confidence levels $\delta_j\in(0,1)$ and $\delta_L\in(0,1)$, there exist scalars
$\beta_j(\delta_j)>0$ and $\beta_L(\delta_L)>0$ such that, for all $z\in S_X$,
\begin{align*}
\mathbb{P}\!\left(\|f_{\mathrm{uk},j}(\bar q_j)-\mu_j(\bar q_j)\|\le \beta_j(\delta_j)\|\Sigma_j^{1/2}(\bar q_j)\|\right)\ge \delta_j,\ \ \forall j,\\
\mathbb{P}\!\left(\|f_{\mathrm{uk},L}(q_L)-\mu_L(q_L)\|\le \beta_L(\delta_L)\|\Sigma_L^{1/2}(q_L)\|\right)\ge \delta_L.
\end{align*}
Here $\beta_j(\delta_j)\in\mathbb{R}^{9}$ and $\beta_L(\delta_L)\in\mathbb{R}^{6}$ are vectors with components
\begin{align}
(\beta_j(\delta_j))_i
:= \sqrt{\,2B_j^2 + 300\,\gamma_{j,i}\,
\ln^3\!\Big(\frac{N_j+1}{1-\delta_j^{1/9}}\Big)\,}, \label{eq:beta_agents}\\
(\beta_L(\delta_L))_i
:= \sqrt{\,2B_L^2 + 300\,\gamma_{L,i}\,
\ln^3\!\Big(\frac{N_L+1}{1-\delta_L^{1/6}}\Big)\,}, \label{eq:beta_payload}
\end{align}
where $B_j$ and $B_L$ are the RKHS norm bounds from Assumption~3, $N_j$ and $N_L$ denote the
number of training points in the corresponding frozen datasets, and
$\gamma_{j,i},\gamma_{L,i}$ are the (maximum) information gains associated with the chosen kernels
(defined as in \cite{Srinivas2012}).
\end{lemma}

\begin{proof}
This is a direct consequence of the standard GP uniform error bound for RKHS-bounded functions
(see, e.g., \cite[Thm.~6]{Srinivas2012}), applied componentwise to the agent outputs ($9$ scalars)
and payload outputs ($6$ scalars), and then stacked into the vector inequalities 
\eqref{eq:ass2_agent}--\eqref{eq:ass2_payload}.
\end{proof}

\begin{remark}
Lemma~\ref{lem:gp_error_bound} implies Assumption~2 by setting
$\bar\rho_j(\bar q_j;\delta_j):=\|\beta_j(\delta_j)^\top \Sigma_j^{1/2}(\bar q_j)\|$ and
$\bar\rho_L(q_L;\delta_L):=\|\beta_L(\delta_L)^\top \Sigma_L^{1/2}(q_L)\|$ on $S_{\mathcal X}$.
\end{remark}

\section{Learning-based Tracking Control} \label{sec:tracking}
\subsection{Control objective and leader--follower architecture}
\label{subsec:tracking_objective}

We consider cooperative transport of a rigid payload by a team of $N$ aerial manipulators subject to the dynamics~\eqref{eq:dae_compact}. The control goal is inherently \emph{hierarchical}: at the upper level, we seek to track a smooth, time-varying reference trajectory for the payload pose $(p_L(t),R_L(t))\in SE(3)$ with bounded tracking errors in the presence of model uncertainties, external disturbances, and manipulator--payload coupling effects. At the lower level, we require each aerial manipulator to robustly realize the assigned contact forces and the induced base motion imposed by the rigid grasp constraints, so that the applied contact forces remain close to the commanded ones despite agent-level uncertainties.

This two-layer viewpoint is essential because the payload wrench is generated only through contact interactions: while the allocation layer can compute commanded contact forces consistent with the grasp map, the full DAE is driven by the forces actually applied by the agents. By explicitly enforcing constraint-consistent realization at the agent level, we prevent disturbances and unmodeled dynamics of individual manipulators from corrupting the net payload wrench via realization errors, thereby enabling a cascade analysis of the coupled closed-loop system.

The payload dynamics evolve on $SE(3)$ and require control of a six-dimensional wrench consisting of a resultant force and moment.
Individual aerial manipulators do not act directly on the payload wrench; instead, they generate contact forces at fixed attachment points, which collectively determine the payload motion through the grasp matrix.
As a result, the grasp map is generally non-injective: the same payload wrench can be produced by different contact-force vectors, whose components in $\ker G(R_L)$ represent internal forces that do not affect the payload motion.

To address redundancy in a scalable and physically consistent manner,
we adopt a leader--follower architecture that separates
(i) payload-level wrench generation, (ii) contact-level wrench allocation, and
(iii) agent-level wrench realization.
This yields a well-posed design that avoids unnecessary force conflicts,
facilitates multi-level stability analysis for \eqref{eq:dae_compact}, and enables 
extensions to learning-based compensation.

The \emph{leader} generates a desired payload wrench
$W_{L,d}(t)=\begin{bmatrix} F_{L,d}(t)\\ \tau_{L,d}(t)\end{bmatrix}\in\mathbb R^6$
based on the payload reference trajectory and the measured (or estimated) payload state.
In the nominal case, $W_{L,d}$ is designed using standard geometric techniques on $SE(3)$
and depends on payload tracking errors and their derivatives.
The term ``leader'' refers to the coordination role in wrench generation; the leader agent,
like all agents, must still realize the resulting contact commands through its local inputs. At the contact-allocation level, the team computes commanded contact forces
$\lambda_{\mathrm{cmd}}\in\mathbb R^{3Nn}$ such that $G(R_L)\,\lambda_{\mathrm{cmd}} = W_{L,d}$, where $G(R_L)\in\mathbb R^{6\times 3Nn}$ is the grasp matrix.
Since this equation is generally underdetermined, the team solves a wrench allocation problem
that distributes load among agents while regulating internal forces.

Crucially, the DAE is driven by the \emph{applied} contact forces $\lambda_{\mathrm{app}}$,
which may differ from $\lambda_{\mathrm{cmd}}$ due to agent dynamics and disturbances.
We therefore explicitly account for a realization error
$e_\lambda := \lambda_{\mathrm{app}}-\lambda_{\mathrm{cmd}}$ and the induced wrench mismatch
$\Delta W := G(R_L)\lambda_{\mathrm{app}}-W_{L,d}$.
This motivates learning residual models both at the agent level (to reduce realization error)
and at the payload level (to compensate net unmodeled wrench effects).

Unmodeled effects acting on both the payload and the aerial manipulators are compensated
using the learned residual models introduced in Section~\ref{sec:learning}.
Within the hierarchy above, learning enters as feedforward corrections
both in the leader's wrench generation and in each agent's local realization layer,
without altering the separation between payload objectives and force allocation.
%This preserves geometric and physical consistency while enabling high-probability stability
%guarantees for the full DAE closed loop.

\subsection{Nominal payload wrench tracking}
\label{subsec:payload_wrench_tracking}

We design a payload-level controller that generates a desired wrench
$W_{L,d}=[F_{L,d};\tau_{L,d}]\in\mathbb R^6$ ensuring exponential tracking of the payload position
and almost-global exponential tracking of its attitude in the nominal case.

Neglecting unknown forces and moments, the payload dynamics are
\begin{align}
m_L\dot v_L &= F_L + f_{g,L},\quad\dot p_L = v_L, \label{eq:pl_nom_kin}\\
J_L\dot\omega_L + \omega_L\times(J_L\omega_L) &= \tau_L,
\label{eq:pl_nom_rot}
\end{align}
where $(F_L,\tau_L)$ denotes the net wrench applied to the payload by the aerial
manipulators through the rigid grasp, and $f_{g,L}$ is the gravitational force
expressed in the inertial frame.

\begin{assumption}
\label{ass:wrench_matching_nominal}
In the ideal case of exact wrench realization, the net payload wrench satisfies
$W_L = G(R_L)\lambda_{\mathrm{app}} = W_{L,d}$.
\end{assumption}

Let $(p_{L,d}(t),R_{L,d}(t))\in SE(3)$ be a smooth reference trajectory with bounded
derivatives, and define $v_{L,d}:=\dot p_{L,d}$. Define the translational tracking errors
$e_p := p_L - p_{L,d}$,  $e_v := v_L - v_{L,d}$; and the attitude and angular-velocity errors using the standard $SO(3)$ error maps
$e_R := \frac{1}{2}\big(R_{L,d}^\top R_L - R_L^\top R_{L,d}\big)^\vee$, $e_\omega := \omega_L - R_L^\top R_{L,d}\,\omega_{L,d}$.
We also use the attitude error function
$\Psi(R_L,R_{L,d}) := \frac{1}{2}\tr\!\big(I - R_{L,d}^\top R_L\big)$.

The payload-level controller (leader) generates the desired wrench $W_{L,d}=[F_{L,d};\tau_{L,d}]$ as
\begin{align}
F_{L,d} :=& m_L\ddot p_{L,d} - m_L K_p e_p - m_L K_v e_v - f_{g,L},
\label{eq:pl_Fd}\\
\tau_{L,d} := &-K_R e_R - K_\omega e_\omega
+ \omega_L\times(J_L\omega_L)\nonumber\\
&- J_L\!\left(\hat\omega_L R_L^\top R_{L,d}\omega_{L,d}
- R_L^\top R_{L,d}\dot\omega_{L,d}\right),
\label{eq:pl_taud}
\end{align}
where $K_p,K_v,K_R,K_\omega$ are positive definite gain matrices.

Under Assumption~\ref{ass:wrench_matching_nominal}, substituting
\eqref{eq:pl_Fd}--\eqref{eq:pl_taud} into \eqref{eq:pl_nom_kin}--\eqref{eq:pl_nom_rot}
yields the closed-loop error dynamics
\begin{align}
\dot e_p &= e_v, \quad \dot e_v = -K_p e_p - K_v e_v, \label{eq:pl_cl_ep}\\
\dot e_R &= \mathcal{E}(R_L,R_{L,d})\,e_\omega, \label{eq:pl_cl_eR}\\
J_L\dot e_\omega &= -K_R e_R - K_\omega e_\omega,
\label{eq:pl_cl_eom}
\end{align}
where $\mathcal{E}(\cdot)$ is a smooth, bounded matrix depending on $e_R$.

The translational error subsystem \eqref{eq:pl_cl_ep} 
is exponentially stable.
Likewise, the rotational error subsystem
\eqref{eq:pl_cl_eR}--\eqref{eq:pl_cl_eom} is exponentially stable about
$e_R=e_\omega=0$ in the standard almost-global sense on $SO(3)$ (see e.g., \cite{lee2010geometric}).

%========================================================
\subsection{Nominal tracking control under rigid grasp constraints}
\label{subsec:tracking_nominal}

We now design a nominal tracking controller for the cooperative aerial manipulation
system described in Section IV-\ref{subsec:dae_model}, assuming perfect knowledge of the
dynamics and neglecting the unknown terms $\mathbf f_{uk,j}$ and $\mathbf f_{uk,L}$.

Let $(p_{L,d}(t),R_{L,d}(t))\in SE(3)$ denote a sufficiently smooth desired payload
trajectory, with associated velocities and accelerations
$(v_{L,d},\dot v_{L,d},\omega_{L,d},\dot\omega_{L,d})$, where
$v_{L,d}:=\dot p_{L,d}$ and $\dot v_{L,d}:=\ddot p_{L,d}$.
The control objective is to drive the payload pose to the reference while respecting
the rigid grasp constraints \eqref{eq:dae_constraints} and the underactuated nature
of the aerial vehicles.

%We use the payload tracking errors $(e_p,e_v,e_R,e_\omega)$ defined in
%Section IV-\ref{subsec:payload_wrench_tracking}; in particular, $e_v=v_L-v_{L,d}$.
%We first define a \emph{virtual payload wrench}
%$W_{L,d}=[F_{L,d};\tau_{L,d}]\in\mathbb R^6$
%that would achieve tracking if applied directly to the payload. Specifically, we set
%\begin{align}
%F_{L,d} &:=
%m_L\dot v_{L,d}
%- m_LK_p e_p - m_LK_v e_v - f_{g,L},
%\label{eq:payload_force_cmd}\\
%\tau_{L,d} &:=
%-K_R e_R - K_\omega e_\omega
%+ \omega_L\times(J_L\omega_L)\nonumber\\
%&\quad
%- J_L\!\left(\hat\omega_L R_L^\top R_{L,d}\omega_{L,d}
%- R_L^\top R_{L,d}\dot\omega_{L,d}\right),
%\label{eq:payload_torque_cmd}
%\end{align}
%where $K_p,K_v,K_R,K_\omega$ are symmetric positive definite gain matrices.

%The desired payload wrench $W_{L,d}$ must be realized through the stacked contact forces
%$\lambda\in\mathbb R^{3Nn}$ generated by the aerial manipulators, subject to the grasp relation
%\begin{equation}
%W_L \;=\; G(R_L)\,\lambda,\qquad G(R_L)\in\mathbb R^{6\times 3Nn}.
%\label{eq:grasp_relation}
%\end{equation}
Since the number of contact-force variables generally exceeds the wrench dimension,
the force allocation problem is redundant and $G(R_L)$ admits a nontrivial nullspace.
Consequently, there exist infinitely many contact-force distributions that generate the
same payload wrench. The components of $\lambda$ in $\ker(G(R_L))$ correspond to
\emph{internal forces}, which do not affect the payload motion but influence grasp tensions,
actuator effort, and robustness margins.

%========================================================

%To obtain a structured and scalable control architecture, we adopt a leader--follower strategy
%separating (i) payload-level wrench generation, (ii) contact-level wrench allocation, and
%(iii) agent-level realization. One subset of contacts (the \emph{leader block}) is designated
%to ensure realization of $W_{L,d}$, while the remaining contacts (the \emph{follower block})
%regulate internal forces without altering the net payload wrench.

Partition the grasp matrix as
$G(R_L)=\begin{bmatrix}G_\ell(R_L) & G_f(R_L)\end{bmatrix}$, where
$n_\ell\in\{1,\ldots,Nn\}$ denotes the number of leader contacts (i.e., the number of grasp points assigned to the leader block),
so that $G_\ell(R_L)\in\mathbb R^{6\times 3n_\ell}$ collects the columns associated with the leader contacts and
$G_f(R_L)\in\mathbb R^{6\times 3(Nn-n_\ell)}$ collects the remaining columns.
Assume the leader grasp block has full row rank, $\mathrm{rank}(G_\ell(R_L))=6$.
Accordingly, the leader contact forces are chosen as
\begin{equation}
\lambda_\ell \;=\; G_\ell^\dagger(R_L)\,W_{L,d},
\label{eq:leader_force}
\end{equation}
where $G_\ell^\dagger$ denotes a right pseudoinverse chosen so that
$G_\ell(R_L)G_\ell^\dagger(R_L)=I_6$.
The follower forces are restricted to internal-force directions:
\begin{equation}
\lambda_f \;=\; \lambda_{f,0} + N\!\big(G_f(R_L)\big)\,\eta,
\,\,
\lambda_{f,0}\in\ker\!\big(G_f(R_L)\big),
\label{eq:follower_force}
\end{equation}
where $N(G_f(R_L))$ is a basis matrix for $\ker(G_f(R_L))$ and $\eta$ is a free internal-force input.

\begin{lemma}
\label{lem:wrench_consistency}
Assume $\mathrm{rank}(G_\ell(R_L))=6$ and choose the allocation \eqref{eq:leader_force}--\eqref{eq:follower_force}
with $\lambda_{f,0}\in\ker(G_f(R_L))$. Then the resulting net payload wrench satisfies
\begin{equation}
W_L=G(R_L)\lambda=G_\ell(R_L)\lambda_\ell + G_f(R_L)\lambda_f=W_{L,d},
\label{eq:wrench_consistency}
\end{equation}
independently of the internal-force input $\eta$.
\end{lemma}

\begin{proof}
By the full row-rank assumption, choose $G_\ell^\dagger$ such that $G_\ell G_\ell^\dagger=I_6$.
Then \eqref{eq:leader_force} gives $G_\ell(R_L)\lambda_\ell = W_{L,d}$.
Moreover, \eqref{eq:follower_force} with $\lambda_{f,0}\in\ker(G_f(R_L))$ yields
$G_f(R_L)\lambda_f = G_f(R_L)\lambda_{f,0}+G_f(R_L)N(G_f(R_L))\eta = 0$ since
$G_f(R_L)N(G_f(R_L))=0$ by definition of a nullspace basis.
Hence $W_L = G_\ell\lambda_\ell + G_f\lambda_f = W_{L,d}$.\hfill$\square$
\end{proof}

%========================================================

Let the holonomic rigid-grasp constraints be stacked as $\Phi(z)=0$ as in \eqref{eq:dae_constraints}, and define the
constraint manifold $\mathcal M:=\{z\in\mathcal X^N\times\mathcal Q_L:\Phi(z)=0\}$.

\begin{prop}
\label{prop:constraint_invariance}
Assume the DAE \eqref{eq:dae_compact} is \emph{regular} on $\mathcal M$, in the sense that for every
$(z,u)\in\mathcal M\times\mathcal U$ there exists a locally unique multiplier
$\lambda=\lambda(z,u)$ such that the corresponding solution satisfies $\frac{d}{dt}\Phi(z(t))=0$.
Consider the closed loop obtained by applying the wrench command
\eqref{eq:pl_Fd}--\eqref{eq:pl_taud} together with the allocation
\eqref{eq:leader_force}--\eqref{eq:follower_force} and any smooth agent-level inputs $u(t)$ compatible with
\eqref{eq:dae_compact}. If the initial condition is consistent, i.e.,
$z(0)\in\mathcal M$ and $\frac{d}{dt}\Phi(z(t))\vert_{t=0}=0$,
then the corresponding maximal solution remains on the constraint manifold:
$z(t)\in\mathcal M$ for all $t$ in its interval of existence.
\end{prop}

\begin{proof}
Along any differentiable trajectory $z(t)$, differentiating the holonomic constraint gives
\[
\frac{d}{dt}\Phi(z(t)) \;=\; D\Phi(z(t))\,\dot z(t).
\]
Substituting the DAE dynamics yields the \emph{constraint-consistency equation}
\[
0 \;=\; D\Phi(z)\big(F(z,u,\lambda)+F_{uk}(z)\big).
\]
By the regularity assumption, for each $(z,u)\in\mathcal M\times\mathcal U$ there exists a locally unique
$\lambda=\lambda(z,u)$ such that the above identity holds, hence
$\frac{d}{dt}\Phi(z(t))=0$ along the resulting solution. Therefore, $\Phi(z(t))$ is constant on the interval of existence.
Since $z(0)\in\mathcal M$ implies $\Phi(z(0))=0$, we obtain $\Phi(z(t))\equiv 0$ for all $t$, i.e., $z(t)\in\mathcal M$.\hfill$\square$
\end{proof}

%========================================================

Given the commanded contact forces $\lambda_{j,b}$, each UAV agent $j$ must generate the corresponding body thrust $u_j$
and torques $(\tau_j,\tau_{r,j})$ in a manner consistent with the dynamics
\eqref{eq:dae_agent_R}--\eqref{eq:dae_agent_mom}. The translational objective is implicit: due to the rigid grasp constraints
\eqref{eq:dae_constraints}, the agents are kinematically forced to follow the payload motion. Accordingly, the individual
agent controllers are designed to track the induced base trajectories and to realize the commanded contact forces while
respecting thrust-direction underactuation inherent to multirotor platforms.

%To capture the fact that the DAE is driven by the \emph{applied} contact forces, we denote by
%$\lambda_{\mathrm{app}}\in\mathbb R^{3Nn}$ the stacked vector of forces actually generated at the contacts and define
%\begin{equation}
%e_\lambda \;:=\; \lambda_{\mathrm{app}}-\lambda_{\mathrm{cmd}},
%\quad
%\Delta W \;:=\; G(R_L)\lambda_{\mathrm{app}}-W_{L,d},
%\label{eq:realization_error_defs_nom}
%\end{equation}
%so that $\Delta W = G(R_L)e_\lambda$.

\begin{assumption}
\label{ass:agent_layer_iss}
On each interval $t\in[t_k,t_{k+1})$, the agent-level controllers use the frozen GP
oracles $\hat{\mathbf f}_{uk,j,k}=\mu_{j,k}$ trained on $\mathcal D_{j,k}$ and yield a bound
on the induced payload wrench mismatch $\Delta W$ of the form
\begin{align}
\|\Delta W(t)\|
\le&
\alpha_W\,e^{-\gamma_W(t-t_k)}\|e_\lambda(t_k)\|
+\vartheta_W\nonumber\\
&+\sum_{j=1}^N \kappa_{j}\,\bar\rho_j(\bar q_j(t);\delta_j),
\label{eq:DeltaW_bound_agentlayer}
\end{align}
for all consistent trajectories in $S_{\mathcal X}$, where $\alpha_W,\gamma_W,\kappa_j\ge 0$ are constants and
$\vartheta_W\ge 0$ collects bounded effects due to measurement noise, numerical differentiation, and contact-force estimation.
\end{assumption}

\begin{remark}
Under Lemma~\ref{lem:gp_error_bound}, one may choose
$\bar\rho_j(\bar q_j;\delta_j)=\beta_j(\delta)\|\Sigma_{j,k}^{1/2}(\bar q_j)\|$ on $S_{\mathcal X}$.
Hence \eqref{eq:DeltaW_bound_agentlayer} makes the payload-level disturbance induced by
agent uncertainties explicit in terms of the agent GP covariances.
\end{remark}

Under the wrench command \eqref{eq:pl_Fd}--\eqref{eq:pl_taud},
the allocation \eqref{eq:leader_force}--\eqref{eq:follower_force}, and Proposition~\ref{prop:constraint_invariance},
solutions remain on $\mathcal M$ and the nominal allocation ensures $W_L=W_{L,d}$ whenever $\Delta W\equiv 0$
(Lemma~\ref{lem:wrench_consistency}). In practice, the applied contact forces induce a mismatch $\Delta W$ which is treated through the interface \eqref{eq:DeltaW_bound_agentlayer} and will be
propagated to payload tracking performance in the learning-based analysis below.

%========================================================
\subsection{Learning-based tracking control}
\label{subsec:tracking_learning}

We now extend the leader--follower tracking controller of
Section IV-\ref{subsec:tracking_nominal} to the setting where both the UAV agents and the payload
are subject to unknown and unmodeled dynamics, as introduced in Section~\ref{sec:learning}.
The goal is to preserve bounded tracking performance while explicitly accounting for learning
uncertainty in a way that is compatible with the DAE structure \eqref{eq:dae_compact} and the
hierarchical realization.

%========================================================

At the payload layer, the leader augments the nominal wrench command with GP-based feedforward
compensation of the unknown payload-level force/moment terms. At the agent layer, each UAV uses
its own GP model to compensate local unknown dynamics so as to improve contact-force realization.
This is consistent with the fact that (i) payload disturbances alone induce disturbances on the
agents through the rigid constraints, and (ii) each agent has additional, local
uncertainties. Both effects are represented in the coupled model and are addressed by learning.

To keep the analysis modular, we treat the agent layer through an input-to-state interface:
agent-level uncertainties affect the payload only via the wrench mismatch $\Delta W$ induced by
the contact-force realization error $e_\lambda$.
The payload-level GP then compensates residual unknown wrench terms acting directly on the payload.

%========================================================

Let $\hat{\mathbf f}_{uk,L}(q_L)=\mu_{L,k}(q_L)$ denote the frozen GP posterior mean on
$t\in[t_k,t_{k+1})$, with covariance $\Sigma_{L,k}(q_L)$, where
$q_L=(p_L,v_L,R_L,\omega_L)$.
We define the learning-augmented virtual payload wrench $W_{L,d}=[F_{L,d};\tau_{L,d}]$ as
\begin{align}
F_{L,d}
:=&
m_L\dot v_{L,d}
- m_LK_p e_p - m_LK_v e_v\nonumber\\
&- f_{g,L}
- \hat{\mathbf f}^{p}_{uk,L,k}(q_L),
\label{eq:payload_force_cmd_learning}\\
\tau_{L,d}
:=&
-K_R e_R - K_\omega e_\omega
+ \omega_L\times(J_L\omega_L)- \hat{\mathbf f}^{\omega}_{uk,L,k}(q_L)\nonumber\\
&
- J_L\!\left(\hat\omega_L R_L^\top R_{L,d}\omega_{L,d}
- R_L^\top R_{L,d}\dot\omega_{L,d}\right),
\label{eq:payload_torque_cmd_learning}
\end{align}
where $K_p,K_v,K_R,K_\omega$ are symmetric positive definite gain matrices and the learned
terms are evaluated at the current payload state.
The allocation layer remains unchanged: $W_{L,d}$ is mapped to commanded contact forces
$\lambda_{\mathrm{cmd}}$ via \eqref{eq:leader_force}--\eqref{eq:follower_force}.

Define the learning errors (frozen on $[t_k,t_{k+1})$) as
\[
\tilde{\mathbf f}_{L,k}(q_L)
:=
\mathbf f_{uk,L}(q_L)-\hat{\mathbf f}_{uk,L,k}(q_L)
=
\begin{bmatrix}\tilde{\mathbf f}^p_{L,k}(q_L)\\ \tilde{\mathbf f}^\omega_{L,k}(q_L)\end{bmatrix}.
\]
Substituting \eqref{eq:payload_force_cmd_learning}--\eqref{eq:payload_torque_cmd_learning} and
$\begin{bmatrix}F_L\\ \tau_L\end{bmatrix}
=
G(R_L)\lambda_{\mathrm{app}}
=
W_{L,d} + \Delta W$ into the payload error dynamics yields
\begin{align}
m_L\dot e_v
&=
- m_LK_p e_p - m_LK_v e_v
+ \Delta F
+ \tilde{\mathbf f}^{p}_{L,k}(q_L),
\label{eq:err_ev_learning}\\
J_L\dot e_\omega
&=
- K_R e_R - K_\omega e_\omega
+ \Delta \tau
+ \tilde{\mathbf f}^{\omega}_{L,k}(q_L),
\label{eq:err_eomega_learning}
\end{align}
where $\Delta W=[\Delta F;\Delta \tau]$.

%========================================================

We now establish a tracking guarantee that explicitly accounts for (i) learning uncertainty
through $\Sigma_{L,k}$ and (ii) the agent-layer realization mismatch through $\Delta W$.

\begin{theorem}
\label{thm:learning_tracking}
Suppose Assumptions~\ref{ass:dataset}, \ref{ass:oracle_bound}, and \ref{ass:rkhs} hold.
Consider the learning-augmented payload wrench
\eqref{eq:payload_force_cmd_learning}--\eqref{eq:payload_torque_cmd_learning}
combined with the allocation \eqref{eq:leader_force}--\eqref{eq:follower_force}.
Assume the agent-layer interface in Assumption~\ref{ass:agent_layer_iss} holds.

Then, for any confidence level $\delta\in(0,1)$, the payload tracking errors
$(e_p,e_v,e_R,e_\omega)$ are uniformly ultimately bounded with probability at least $\delta$ on $S_{\mathcal X}$. Moreover, up to a decaying transient from \eqref{eq:DeltaW_bound_agentlayer},
the ultimate bound scales proportionally to
$\displaystyle{
\bar\rho_L(q_L;\delta)
\;+\;
\sum_{j=1}^N \bar\rho_j(\bar q_j;\delta)
\;+\;\vartheta_W}$. \end{theorem}

\begin{proof}
% --- Replace the confidence event definition by a multiline/aligned version ---
Fix an update interval $t\in[t_k,t_{k+1})$ and split the global confidence level $\delta$
into $\delta_L,\delta_1,\dots,\delta_N$. Define the confidence event
\begin{equation}
\Omega_\delta :=
\Big\{
\begin{aligned}
&\|\tilde f_{L,k}(q_L(t))\|\le \bar\rho_L(q_L(t);\delta_L),\\
&\|f_{\mathrm{uk},j}(\bar q_j(t))-\hat f_{\mathrm{uk},j,k}(\bar q_j(t))\|
\le \bar\rho_j(\bar q_j(t);\delta_j)
\end{aligned}
\Big\} 
\label{eq:Omega_delta_union}
\end{equation}  $\forall j,\ \forall t\ge 0$. 
By Assumption~2 and the union bound (Remark~1), $\mathbb{P}(\Omega_\delta)\ge \delta$ on $S_X$.

Consider the Lyapunov function
\begin{align*}
V(e_p,e_v,e_R,e_\omega)
:= \tfrac12 m_L\|e_v\|^2 + \tfrac12 m_L e_p^\top K_p e_p\\
+ \tfrac12 e_\omega^\top J_L e_\omega + \tfrac12 e_R^\top K_R e_R .\end{align*}

Along trajectories on the constraint manifold $\mathcal M$ we have
$\dot e_p=e_v$ and $\dot e_R=\mathcal E(R_L,R_{L,d})e_\omega$ for a smooth bounded matrix
$\mathcal E(\cdot)$. Using \eqref{eq:err_ev_learning}--\eqref{eq:err_eomega_learning} and cancelling cross terms yields
\begin{equation}
\label{eq:Vdot_bound_38}
\begin{aligned}
\dot V
=& -m_L e_v^\top K_v e_v - e_\omega^\top K_\omega e_\omega
\\&+ e_v^\top(\Delta F+\tilde{\mathbf f}^p_{L,k})
+ e_\omega^\top(\Delta\tau+\tilde{\mathbf f}^\omega_{L,k})\\
\le &-k_v\|e_v\|^2 - k_\omega\|e_\omega\|^2\\
&+ \|(e_v,e_\omega)\|\Big(\|\Delta W\|+\|\tilde{\mathbf f}_{L,k}(q_L)\|\Big),
\end{aligned}
\end{equation}
where $k_v:=m_L\lambda_{\min}(K_v)$ and $k_\omega:=\lambda_{\min}(K_\omega)$.

On $\Omega_\delta$, $\|\tilde{\mathbf f}_{L,k}(q_L)\|\le \bar\rho_L(q_L;\delta)$, hence
\begin{equation}
\label{eq:Vdot_bound_k}
\dot V \le -k\|(e_v,e_\omega)\|^2
+ \|(e_v,e_\omega)\|\Big(\|\Delta W\|+\bar\rho_L(q_L;\delta)\Big),
\end{equation}
with $k:=\min\{k_v,k_\omega\}$.

Applying Young's inequality, for any $\varepsilon\in(0,k)$,
\begin{equation}
\label{eq:Young_39}
\begin{aligned}
\|(e_v,e_\omega)\|\Big(\|\Delta W\|+&\bar\rho_L(q_L;\delta_L)\Big)
\le \varepsilon\|(e_v,e_\omega)\|^2\\
&+\frac{1}{4\varepsilon}
\Big(\|\Delta W\|+\bar\rho_L(q_L;\delta_L)\Big)^2.
\end{aligned}
\end{equation}

Combining \eqref{eq:Vdot_bound_k} and \eqref{eq:Young_39} gives
\begin{equation}
\label{eq:Vdot_bound_40}
\begin{aligned}
\dot V
\le -(k-&\varepsilon)\|(e_v,e_\omega)\|^2\\
&\quad+\frac{1}{4\varepsilon}
\Big(\|\Delta W\|+\bar\rho_L(q_L;\delta_L)\Big)^2,
\, \text{on }\Omega_\delta .
\end{aligned}
\end{equation}

Next we use the agent-layer interface (Assumption~\ref{ass:agent_layer_iss}) on $[t_k,t_{k+1})$. Define the \emph{decaying transient}
\[
\chi_k(t)\;:=\;\alpha_W e^{-\gamma_W(t-t_k)}\|e_\lambda(t_k)\|,\qquad t\in[t_k,t_{k+1}),
\]
and the aggregated uncertainty level
\[
\Xi(t) := \bar\rho_L(q_L(t);\delta_L) + \vartheta_W +
\sum_{j=1}^N \kappa_j\,\bar\rho_j(\bar q_j(t);\delta_j) + \chi_k(t).\]
Then $\|\Delta W(t)\|+\bar\rho_L(q_L(t);\delta_L)\le \Xi(t)$ and \eqref{eq:Vdot_bound_k} implies
\begin{equation}
\dot V
\le
-(\underline{k}-\varepsilon)\,\|(e_v,e_\omega)\|^2
+\frac{1}{4\varepsilon}\,\Xi(t)^2
\qquad\text{on } \Omega_\delta.
\label{eq:Vdot_step3}
\end{equation}

Finally, relate $\|(e_v,e_\omega)\|^2$ to $V$. Since $V\ge \tfrac12 m_L\|e_v\|^2+\tfrac12\lambda_{\min}(J_L)\|e_\omega\|^2$,
there exists $c_1>0$ such that $\|(e_v,e_\omega)\|^2\ge c_1 V$ for all $(e_p,e_v,e_R,e_\omega)$.
Thus, for some constants $c_2,c_3>0$, $\dot V \le -c_2 V + c_3\,\Xi(t)^2$ on $\Omega_\delta$.

By the comparison lemma, for all $t\in[t_k,t_{k+1})$,
\[
V(t)\le e^{-c_2(t-t_k)}V(t_k)+c_3\int_{t_k}^{t} e^{-c_2(t-s)}\Xi(s)^2\,ds.
\]
Since $S_{\mathcal X}$ is compact, $\bar\rho_L(\cdot;\delta)$ and $\bar\rho_j(\cdot;\delta)$ are bounded on $S_{\mathcal X}$,
and $\chi_k(t)$ decays exponentially. Therefore the integral term is uniformly bounded and yields a uniform ultimate bound.
Moreover, up to the exponentially decaying transient $\chi_k(t)$, the ultimate bound scales with
$\bar\rho_L(q_L;\delta_L)+\sum_{j=1}^N\bar\rho_j(\bar q_j;\delta_j)+\vartheta_W$. \hfill$\square$
\end{proof}

\begin{remark}
\label{rem:gp_covariance_bound}
Under Lemma~1, on the confidence event $\Omega_\delta$ one may choose $\bar\rho_L(q_L;\delta_L)=\beta_L(\delta_L)\|\Sigma^{1/2}_{L,k}(q_L)\|$ and  $\bar\rho_j(\bar q_j;\delta_j)=\beta_j(\delta_j)\|\Sigma^{1/2}_{j,k}(\bar q_j)\|$, $\forall j$, for all $q_L,\bar q_j\in S_X$.
Hence, up to the decaying transient in \eqref{eq:DeltaW_bound_agentlayer}, the ultimate tracking accuracy
in Th.~\ref{thm:learning_tracking} scales with the GP predictive uncertainties at the payload and agent states $\displaystyle{\beta_L(\delta)\big\|\Sigma_{L,k}^{1/2}(q_L)\big\|
+\sum_{j=1}^N \beta_j(\delta)\big\|\Sigma_{j,k}^{1/2}(\bar q_j)\big\|
+\vartheta_W}$, and improves monotonically as the covariances contract with additional data.
\end{remark}

\begin{remark}The proposed learning-based controller admits a clear robustness--performance trade-off.
The GP mean provides feedforward compensation of unknown dynamics, while the predictive covariance quantifies the remaining
model mismatch and determines the high-probability tracking accuracy. In parallel, improved contact-force realization reduces
$\Delta W$, tightening the ultimate bound without altering the leader--follower structure.\hfill$\diamond$\end{remark}

%--------------------------------------------------------
\begin{corollary}
\label{cor:asymptotic_recovery}
Under the assumptions of Theorem~\ref{thm:learning_tracking}, fix $\delta\in(0,1)$ and split it into
$\delta_L,\delta_1,\dots,\delta_N$ as in Remark \ref{rem1}.
If, along a sequence of learning updates $k\to\infty$, the predictive covariance satisfies
$\displaystyle{\sup_{q_L\in S_{\mathcal X}}}\big\|\Sigma_{L,k}^{1/2}(q_L)\big\|\to0$, and the agent-layer mismatch satisfies
$\displaystyle{\sup_{t\ge 0}}\|\Delta W_k(t)\|\to0$, then the ultimate bound in Theorem~\ref{thm:learning_tracking} converges to zero on $\Omega_\delta$.
Consequently, in the limit $k\to\infty$ the learning-based closed loop recovers the nominal asymptotic tracking behavior.
\end{corollary}

\begin{proof}
Fix $\delta\in(0,1)$ and work on the event $\Omega_\delta$ (so that all GP error bounds invoked in
Theorem~\ref{thm:learning_tracking} hold uniformly on $S_{\mathcal X}$, with confidence levels
$\delta_L$ for the payload and $\delta_j$ for the agents).
From Theorem~\ref{thm:learning_tracking}, there exist class-$\mathcal{K}\mathcal{L}$ and class-$\mathcal{K}$ functions
$\beta(\cdot,\cdot)$ and $\gamma(\cdot)$ and a nonnegative term $\chi_k(t)$ (the decaying transient from
\eqref{eq:DeltaW_bound_agentlayer}) such that, for all $t\ge 0$,
\begin{align}
\|(e_p,e_v,e_R,e_\omega)(t)\|
\le&
\beta\!\left(\|(e_p,e_v,e_R,e_\omega)(0)\|,t\right)\nonumber\\
&+\gamma\!\left(\bar d_k\right)
+\chi_k(t),
\label{eq:UUB_bound_template}
\end{align}
where the effective disturbance level $\bar d_k$ can be chosen as
\begin{equation}
\bar d_k
\;:=\;
\sup_{q_L\in S_{\mathcal X}}\bar\rho_{L,k}(q_L;\delta_L)
\;+\;
\sup_{t\ge 0}\|\Delta W_k(t)\|
\;+\;\vartheta_W,
\label{eq:dk_def}
\end{equation}
(here $\bar\rho_{L,k}$ denotes the payload GP high-probability error radius on update interval $k$).

By Lemma~\ref{lem:gp_error_bound}, on $\Omega_\delta$ we may take
$\bar\rho_{L,k}(q_L;\delta_L)=\beta_L(\delta_L)\|\Sigma_{L,k}^{1/2}(q_L)\|$, hence
\[
\sup_{q_L\in S_{\mathcal X}}\bar\rho_{L,k}(q_L;\delta_L)
\;\le\;
\beta_L(\delta_L)\sup_{q_L\in S_{\mathcal X}}\|\Sigma_{L,k}^{1/2}(q_L)\|
\;\longrightarrow\;0.
\]
By assumption, $\sup_{t\ge 0}\|\Delta W_k(t)\|\to 0$ as $k\to\infty$.
Therefore, from \eqref{eq:dk_def} we obtain $\bar d_k\to \vartheta_W$.
If $\vartheta_W=0$ (idealized case with no measurement/estimation artifacts), then $\bar d_k\to 0$ and
\eqref{eq:UUB_bound_template} implies that the ultimate bound $\gamma(\bar d_k)$ converges to zero, i.e., the closed loop
recovers the nominal asymptotic tracking behavior as $k\to\infty$.

If $\vartheta_W>0$, then the ultimate bound converges to $\gamma(\vartheta_W)$, i.e., performance is limited only by the
non-learning artifacts captured by $\vartheta_W$.\hfill$\square$
\end{proof}
%========================================================
\section{Simulation Results}
\label{sec:simulations}

Next, we validate the hierarchical learning-based leader--follower architecture for cooperative rigid-payload transport.
The simulations are designed to: (i) verify payload tracking under rigid grasp constraints using the coupled agent--payload DAE \eqref{eq:dae_compact},
(ii) isolate the roles of payload-level and agent-level learning within the two-layer realization viewpoint of Section IV-\ref{subsec:tracking_objective}, and
(iii)  connect the observed tracking accuracy to the high-probability UUB bound of Theorem~\ref{thm:learning_tracking}, which depends on
the payload GP uncertainty, the agent GP uncertainties, and the agent-to-payload realization interface through the induced wrench mismatch $\Delta W$.

%\subsection{Simulation setup and performance metrics}
%\label{subsec:sim_setup}

We consider a team of $N=2$ aerial manipulators cooperatively transporting a rigid payload.
Each agent establishes $n=2$ rigid contacts with the payload, yielding a total of $Nn=4$ contacts and a stacked contact-force vector
$\lambda\in\mathbb R^{12}$.
The coupled dynamics are simulated using the full agent--payload DAE model, including constraint reactions and contact wrenches. The payload reference $(p_{L,d}(t),R_{L,d}(t))$ consists of a smooth position trajectory with bounded derivatives and a constant desired attitude.
This excites the translational dynamics while keeping the rotational evaluation interpretable through $\|e_R(t)\|$ and $\|e_\omega(t)\|$.

To emulate realistic mismatch and external effects, we introduce:
(i) parametric mismatch in payload mass and inertia (controller uses nominal $(m_L,J_L)$ different from the true values),
(ii) additive unknown force/torque disturbances acting directly on the payload dynamics, and
(iii) local agent-level unmodeled effects affecting contact-force realization (e.g., unmodeled aerodynamic terms and input biases).
In addition, measurement noise and numerical differentiation are included in the force reconstruction pipeline, producing a bounded realization floor
captured by the constant $\vartheta_W$ in Assumption~\ref{ass:agent_layer_iss}.
All simulations start from consistent initial conditions on the constraint manifold $\mathcal{M}$.

Each agent is assumed to have access to estimates of the contact forces $\lambda_{j,b}$ via model-based force observers (disturbance/unknown-input observers) or residual-based force estimation using onboard state estimates and actuator models.
We denote by $\lambda_{\mathrm{cmd}}$ the commanded contact forces produced by the allocation layer and by $\lambda_{\mathrm{app}}$ the applied/realized forces.
The agent-to-payload interface signal is the induced wrench mismatch $\Delta W(t)$ which drives the payload as an additive disturbance in \eqref{eq:err_ev_learning}--\eqref{eq:err_eomega_learning}.

We compare three controller configurations under identical initial conditions and reference trajectories:

\begin{itemize}
\item Case (C1) Nominal (no learning):
the learning terms are disabled at both layers, i.e., $\hat{\mathbf f}_{uk,L}\equiv 0$ and $\hat{\mathbf f}_{uk,j}\equiv 0$.
\item Case (C2) Payload learning:
the leader uses the payload GP compensation $\hat{\mathbf f}_{uk,L,k}=\mu_{L,k}$ in
\eqref{eq:payload_force_cmd_learning}--\eqref{eq:payload_torque_cmd_learning}, while the agents do not use GP compensation, i.e.,
$\hat{\mathbf f}_{uk,j}\equiv 0$.
\item Case (C3) Payload+agent learning:
both the leader and each agent use their respective GP compensators,
$\hat{\mathbf f}_{uk,L,k}=\mu_{L,k}$ and $\hat{\mathbf f}_{uk,j,k}=\mu_{j,k}$, improving contact-force realization and thereby reducing $\Delta W$.
\end{itemize}

In all cases, the allocation layer is identical and uses \eqref{eq:leader_force}--\eqref{eq:follower_force}.
Thus, differences between (C2) and (C3) reflect the effect of agent-layer learning on the disturbance $\Delta W$.

We report the payload tracking errors $\|e_p(t)\|$, $\|e_v(t)\|$, $\|e_R(t)\|$, and $\|e_\omega(t)\|$,
together with the interface signal $\|\Delta W(t)\|$.
To connect to Theorem~\ref{thm:learning_tracking}, we also plot the predictive standard deviations along the closed-loop trajectories: $\sigma_L(t) := \|\Sigma_{L,k}^{1/2}(q_L(t))\|$, $\sigma_A(t):= \sum_{j=1}^N \|\Sigma_{j,k}^{1/2}(\bar q_j(t))\|$. On the confidence event used in Assumption~\ref{ass:oracle_bound} and Lemma~\ref{lem:gp_error_bound},
the theorem indicates that, up to a decaying transient due to the realization dynamics,
the ultimate tracking accuracy scales with a combination of $\bar\rho_L(q_L;\delta)$, $\sum_j \bar\rho_j(\bar q_j;\delta)$, and $\vartheta_W$.
Accordingly, we examine how reductions in $\sigma_L(t)$ and $\sigma_A(t)$ correlate with reductions in tracking error and  $\|\Delta W(t)\|$.

The GP prior uses a squared-exponential (RBF) covariance with automatic relevance determination (ARD), i.e., one characteristic length-scale per input coordinate, and an additive i.i.d.\ white-noise term. We fit independent scalar GPs per output channel: $6$ GPs for the payload residual (one per wrench component) using the $20$-dimensional payload feature vector, and $9$ GPs per agent for the agent-level residual using the corresponding agent feature vector. In both cases, the kernel hyperparameters (signal variance, ARD length-scales, and noise variance) are learned by maximizing the log marginal likelihood on the associated frozen dataset. Large learned length-scales (near the imposed upper bound) indicate weak sensitivity to the associated feature, whereas smaller values identify the most informative coordinates for predicting the disturbance residual.

Figure~\ref{fig:payload_tracking_se3} reports the position error norm $\|e_p(t)\|$ and the attitude error measure $\Psi(t)$ for (C1)--(C3).
All cases remain bounded and exhibit transient decay followed by nonzero ultimate levels due to residual uncertainties.
Compared with the nominal case (C1), enabling payload learning (C2) reduces the steady tracking offset by compensating the effective payload-level unknown wrench.
Adding agent learning (C3) further tightens the ultimate bound by improving contact-force realization and thereby reducing the induced mismatch $\Delta W$, in agreement with Theorem~\ref{thm:learning_tracking}.
\begin{figure}[h!]
  \centering
  \includegraphics[width=\linewidth]{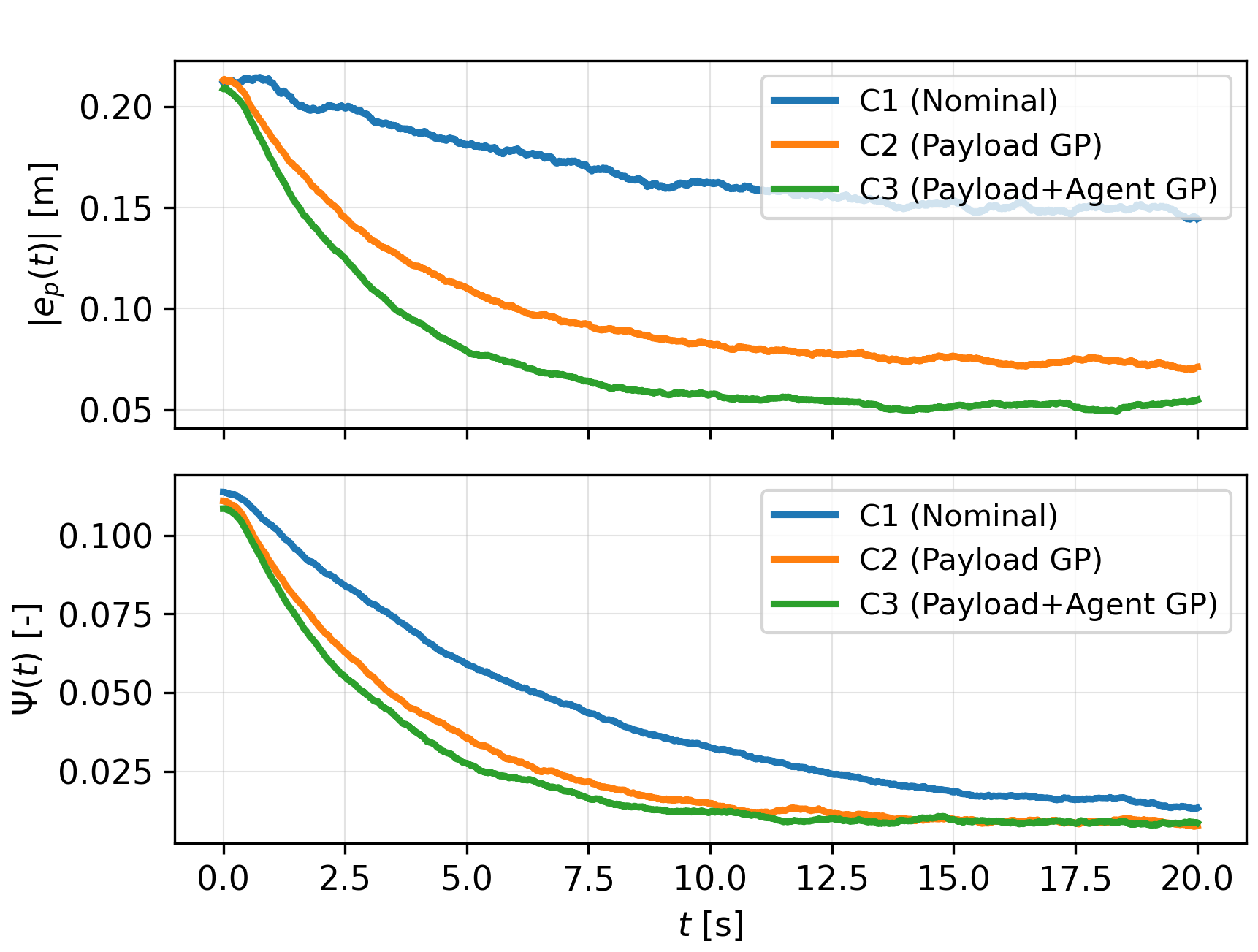}
  \caption{Payload tracking on $SE(3)$ for (C1)--(C3). Top: position error norm $\|e_p(t)\|$. Bottom: attitude error measure $\Psi(t)$. Ultimate bounds decrease from (C1) to (C2) to (C3), consistent with Theorem~\ref{thm:learning_tracking}.}
\label{fig:payload_tracking_se3}
\end{figure}

%\paragraph{Agent-to-payload interface: induced wrench mismatch $\Delta W$.}
Figure~\ref{fig:deltaW} plots $\|\Delta W(t)\|$.
The reduction from (C2) to (C3) highlights the role of agent-layer learning: improving contact-force realization reduces the effective disturbance entering the payload error dynamics.
This observation supports the modular two-layer design, where agent uncertainties affect payload tracking only through $\Delta W$.

\begin{figure}[h!]
  \centering
  \includegraphics[width=\linewidth]{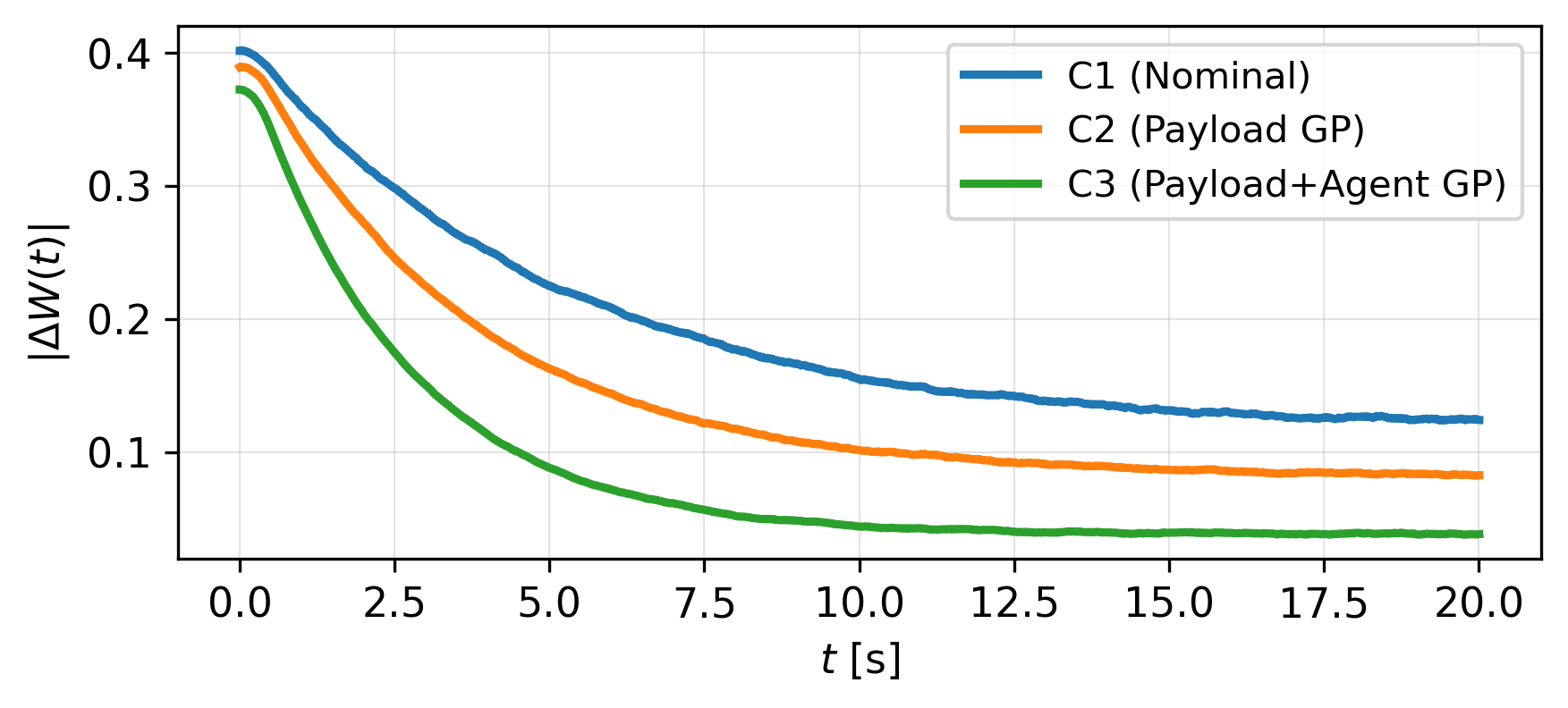}
 \caption{Payload wrench mismatch $\|\Delta W(t)\|$ for (C1)--(C3).
Agent learning reduces the interface disturbance entering the payload error dynamics.}
\label{fig:deltaW}
\end{figure}

%\paragraph{Predictive uncertainty along trajectories.}
Figure~\ref{fig:gp_uncertainty} shows the predictive standard deviations $\sigma_L(t)$ and $\sigma_A(t)$ evaluated along the closed-loop trajectories.
Regions of larger predictive uncertainty correlate with larger transient tracking errors, while better-covered regions (smaller variance) yield tighter tracking.
These trends are consistent with the high-probability residual bounds used in Lemma~\ref{lem:gp_error_bound} and Theorem~\ref{thm:learning_tracking}.

\begin{figure}[h!]
  \centering
  \includegraphics[width=\linewidth]{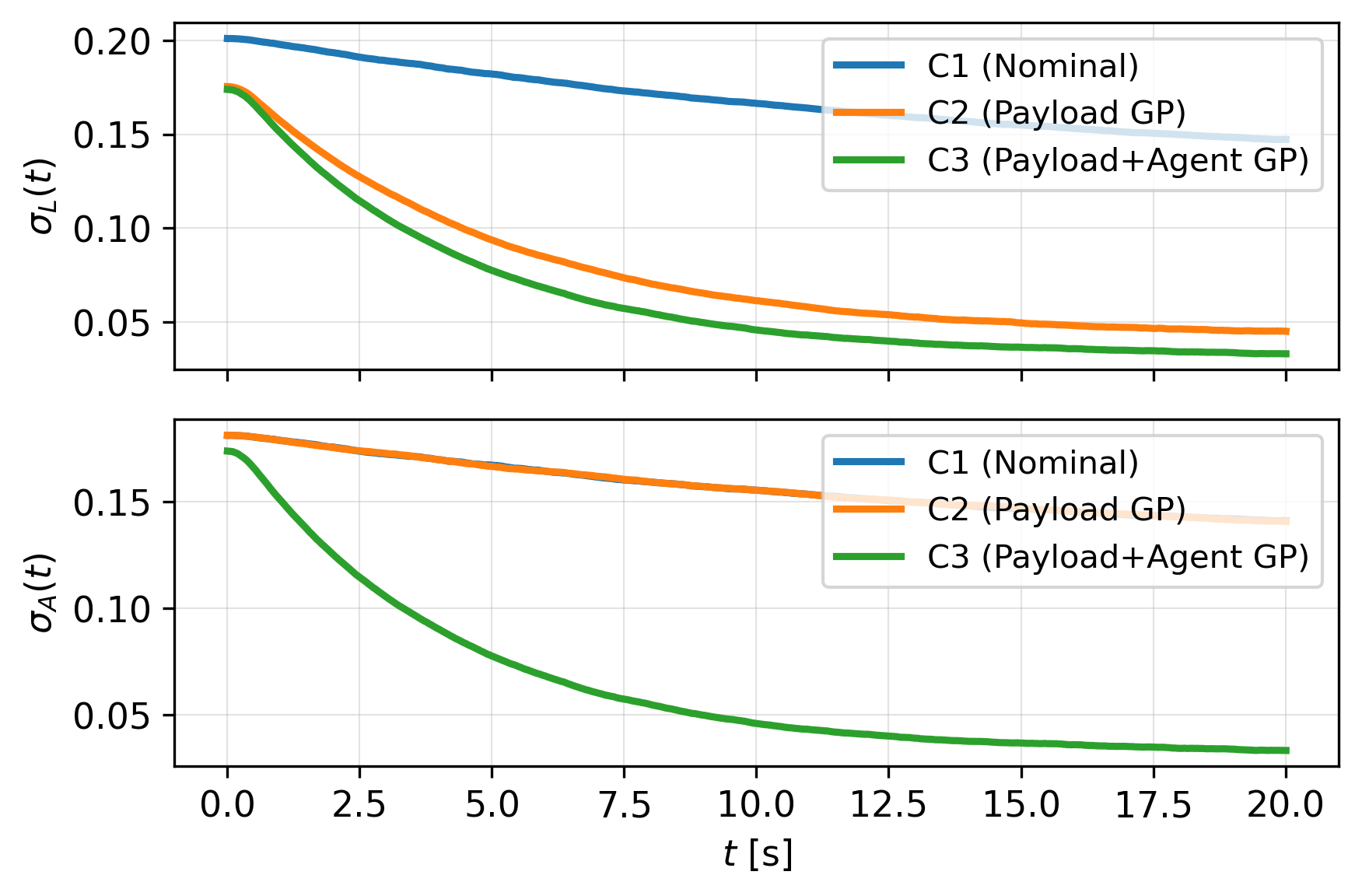}
  \caption{Uncertainty measures along the closed-loop trajectory for (C1)--(C3).
Top: payload-level uncertainty $\sigma_L(t)$ decreases when payload learning is enabled (C2--C3).
Bottom: agent-level uncertainty $\sigma_A(t)$ decreases only when agent learning is enabled (C3); thus $\sigma_A(t)$ remains essentially unchanged from (C1) to (C2).}
  \label{fig:gp_uncertainty}
\end{figure}

%\paragraph{Contact forces and internal-force regulation.}
Figure~\ref{fig:contact_forces} reports the norms of the estimated applied contact forces for each agent.
Across all cases, the force magnitudes remain bounded over the maneuver and do not exhibit high-frequency chattering, indicating that the allocation and realization layers yield a well-conditioned contact-force behavior.
When the internal-force input $\eta$ is activated (see $\|\eta(t)\|$), the contact-force distribution is reshaped along internal-force directions while preserving the commanded payload wrench, in agreement with Lemma~\ref{lem:wrench_consistency}.

%-------------------------------------------------------
% Figures (placeholders)
%-------------------------------------------------------

\begin{figure}[t]
  \centering
  \includegraphics[width=\linewidth]{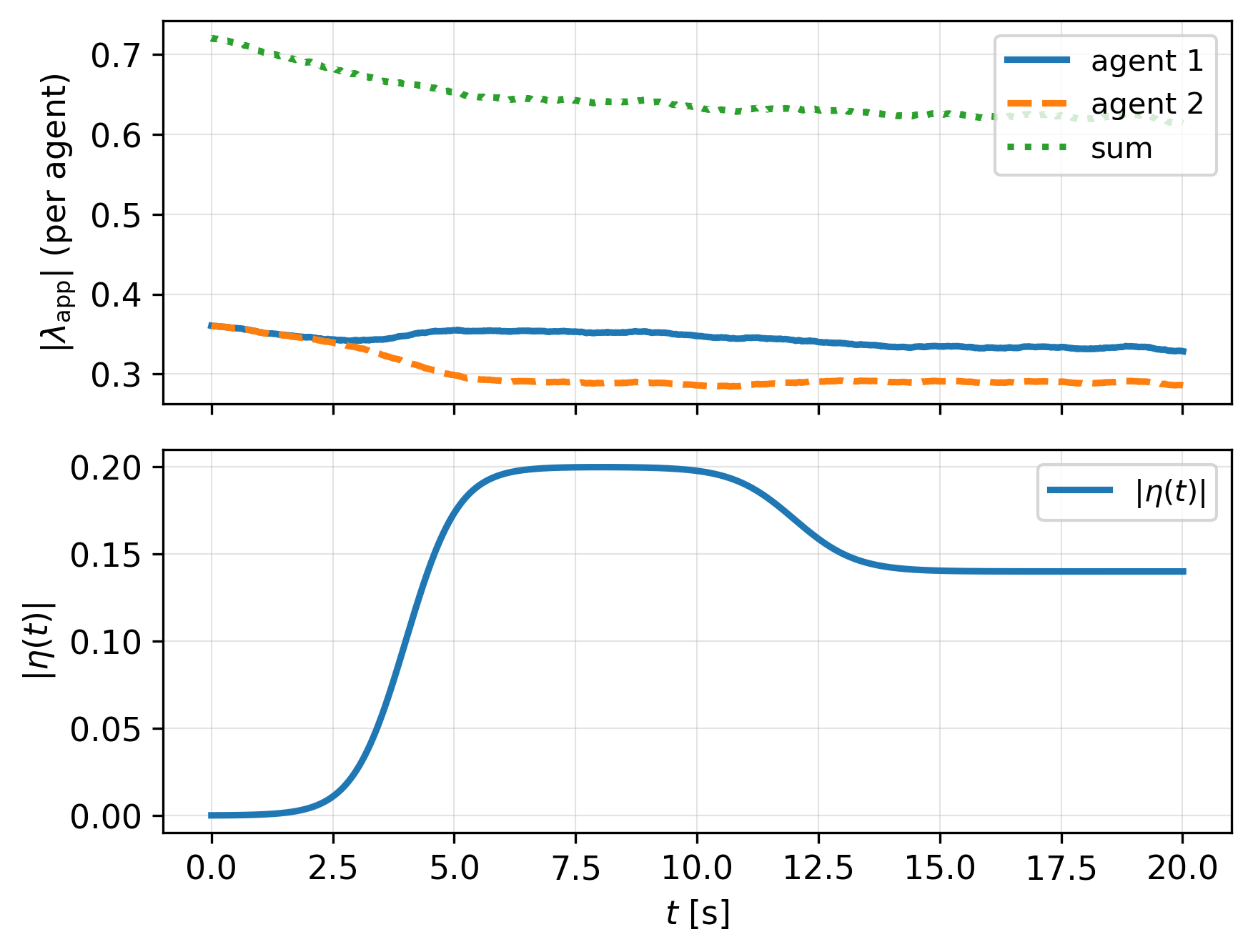}
  \caption{Norms of the estimated applied contact forces for each agent. The forces remain bounded throughout the maneuver and show no high-frequency chattering. The bottom panel shows the internal-force input magnitude $\|\eta(t)\|$; when activated, $\eta$ redistributes contact forces along internal directions without affecting the net payload wrench (Lemma~2).}
  \label{fig:contact_forces}
\end{figure}

Future work will focus on a full experimental implementation on a multi-UAV testbed with rigid grasping, including (i) real-time force reconstruction/estimation for $\lambda_{j,b}$ and online computation of $\Delta W$, (ii) onboard GP training/inference with incremental data management and update scheduling, (iii) allocator and realization layers with actuator/tilt/thrust saturation handling and safety constraints on contact forces and internal-force inputs, and (iv) systematic identification of the disturbance sources (aerodynamics, actuator dynamics, payload parameter drift) to build informative feature vectors and improve generalization across operating conditions. 
\bibliographystyle{IEEEtran}
\bibliography{cluster}

\vspace{12pt}

\end{document}